\documentclass[a4paper,UKenglish,cleveref, autoref, thm-restate, numberwithinsect ]{lipics-v2021}
\usepackage[utf8]{inputenc}

\title{Makespan Scheduling of Unit Jobs with Precedence Constraints in $O(1.995^n)$ time}

\nolinenumbers
\hideLIPIcs

\author{Jesper Nederlof}{Utrecht University, The Netherlands}{j.nederlof@uu.nl}{https://orcid.org/0000-0003-1848-0076}{Supported by the project CRACKNP that has received funding from the European Research Council (ERC) under the European Union’s Horizon 2020 research and innovation programme (grant agreement No 853234)}

\author{C\'eline M. F. Swennenhuis}{Eindhoven University of Technology, The Netherlands}{c.m.f.swennenhuis@tue.nl}{https://orcid.org/0000-0001-9654-8094}{Supported by the Netherlands Organization for Scientific Research under project no. 613.009.031b.}

\author{Karol W\k{e}grzycki}{Saarland University, Saarbrücken, Germany \\
Max Planck Institute for Informatics, Saarbrücken, Germany}{wegrzycki@cs.uni-saarland.de}{0000-0001-9746-5733}{
Supported by the project TIPEA that has   received funding from the European Research Council (ERC) under the European Unions Horizon 2020 research and innovation programme (grant agreement No. 850979).
}

\acknowledgements{The results presented in this paper were obtained during the
trimester on Discrete Optimization at Hausdorff Research Institute for
Mathematics (HIM) in Bonn, Germany. We are thankful for the possibility of
working in the stimulating and creative research environment at HIM. We also
thank Adam Polak for useful discussions.}

\authorrunning{J. Nederlof,  C.\,M.\,F. Swennenhuis and K. W\k{e}grzycki} 

\Copyright{Jesper Nederlof, C\'eline M. F. Swennenhuis, Karol W\k{e}grzycki} 
\date{}

\usepackage{amsmath, amssymb}
\usepackage{xargs}
\usepackage{mathtools}
\usepackage{amsthm}
\usepackage{stmaryrd}
\usepackage[square,sort,comma,numbers]{natbib}
\usepackage{graphicx}
\usepackage{hyperref}
\usepackage{xcolor}
\usepackage{comment}
\usepackage[disable]{todonotes}
\usepackage[linesnumbered,lined,boxed,commentsnumbered,noend]{algorithm2e}
\usepackage{enumerate}
\usepackage{framed}
\usepackage[framemethod=TikZ]{mdframed}

\newtheorem{problem}{Open Problem}
\newtheorem{reduction}{Reduction Rule}
\newtheorem{property}[theorem]{Property}
\newtheorem{hypothesis}[theorem]{Hypothesis}

\newenvironment{insight}
{\mdfsetup{%
    nobreak=true,
	middlelinecolor=gray,
	middlelinewidth=1pt,
	backgroundcolor=gray!10,
    innertopmargin=7pt,
	roundcorner=5pt}
\begin{mdframed}}
{\end{mdframed}}

\newcommand{\eps}{\varepsilon}

\newcommand{\Oh}{\mathcal{O}}
\newcommand{\Os}{\mathcal{O}^{\star}}
\newcommand{\Otilde}{\widetilde{\mathcal{O}}}
\newcommand{\Ot}{\Otilde}

\newcommand{\nat}{\mathbb{N}}

\newcommand{\Gc}{G^{\mathsf{comp}}}
\newcommand{\sched}{$P\, \vert \, \mathrm{prec}, p_j =1 \vert \, C_{\max}$\,}
\newcommand{\AC}{\#\mathrm{AC}}
\newcommand{\DP}{\mathsf{DP}}
\newcommand{\DKS}{D$\kappa$S\xspace}
\newcommand{\den}{\mathrm{den}_{\kappa}}
\DeclarePairedDelimiter{\iverson}{\llbracket}{\rrbracket}

\newcommand{\dual}[1]{\overleftarrow{#1}}

\newcommand{\Pred}{\textsf{pred}}
\newcommand{\Succ}{\textsf{succ}}
\newcommand{\Sinks}{\textsf{sinks}}
\newcommand{\Sources}{\textsf{sources}}
\newcommand{\depth}[2]{d_{#1}(#2)}
\newcommand{\SM}{\textsf{Imp}} 
\newcommand{\SinkVC}{\textsf{Low}} 
\newcommand{\HighVC}{\textsf{High}} 
\newcommand{\EarlyVC}{\textsf{Early}} 

\newcommandx{\unsure}[2][1=]{\todo[linecolor=green,backgroundcolor=green!25,bordercolor=green,#1]{\normalsize #2}}
\newcommandx{\improvement}[2][1=]{\todo[inline,linecolor=blue,backgroundcolor=blue!05,bordercolor=blue,#1]{\normalsize #2}}
\newcommandx{\info}[2][1=]{\todo[linecolor=yellow,backgroundcolor=yellow!25,bordercolor=yellow,#1]{#2}}
\newcommandx{\floatmodel}[2][1=]{\todo[inline,linecolor=red,backgroundcolor=yellow!25,bordercolor=yellow,#1]{#2}}
\newcommandx{\thiswillnotshow}[2][1=]{\todo[disable,#1]{#2}}
\newcommandx{\celine}[2][1=]{\todo[inline,linecolor=green,backgroundcolor=green!25,bordercolor=green,caption={\normalsize \textbf{Celine}},#1]{\normalsize #2}}
\newcommandx{\karol}[2][1=]{\todo[inline,linecolor=blue,backgroundcolor=blue!25,bordercolor=blue,caption={\normalsize \textbf{Karol}},#1]{\normalsize #2}}
\newcommandx{\jesper}[2][1=]{\todo[inline,linecolor=red,backgroundcolor=red!25,bordercolor=red,caption={\normalsize \textbf{Jesper}},#1]{\normalsize #2}}

\begin{document}

\maketitle
\begin{abstract}
In a classical scheduling problem, we are given a set of $n$ jobs of unit length along with precedence constraints and the goal is to find a schedule of these jobs on $m$ identical machines that minimizes the makespan. This problem is well-known to be NP-hard for an unbounded number of machines. Using standard 3-field notation, it is known as $P|\text{prec}, p_j=1|C_{\max}$.  

We present an algorithm for this problem that runs in $\Oh(1.995^n)$ time. Before our work, even for $m=3$ machines the best known algorithms ran in $\Os(2^n)$ time.  In contrast, our algorithm works when the number of machines $m$ is unbounded. A crucial ingredient of our approach is an algorithm with a runtime that is only single-exponential in the vertex cover of the comparability graph of the precedence constraint graph. This heavily relies on insights from a classical result by Dolev and Warmuth (Journal of Algorithms 1984) for precedence graphs without long chains.
\end{abstract}

\clearpage

\section{Introduction}

Scheduling of precedence constrained jobs on identical
machines is a central challenge in the algorithmic study of scheduling problems. In this problem, we
have $n$ jobs, each one of unit length along with $m$ identical
parallel machines on which we can process the jobs. Additionally, the input contains a
set of \emph{precedence constraints} of jobs; a precedence constraint $j' \prec j$ states that job $j'$ has
to be completed before job $j$ can be started. The goal is to schedule the jobs
non-preemptively so that the \emph{makespan} is minimized. Here, the makespan is
the time when the last job is completed. In the \emph{3-field
notation}\footnote{In the 3-field notation, the first entry specifies the type of available
    machine, the second entry specifies the type of jobs, and the last field is the
    objective. In our case, $P$  means that we have identical parallel machines.
    We use $Pm$ to indicate that number of machines is a fixed constant $m$.
    Second entry $\text{prec}, p_j = 1$ indicates that the jobs have precedence
    constraints and unit length. The last field $C_{\max}$ means that the
objective function is to minimize the completion time.}
of Graham~\cite{graham1969bounds} this problem is denoted as \sched.

Despite the extensive interest in the
community~\cite{coffman1972optimal,Gabow:82:An-almost-linear-algorithm,Sethi:76:Scheduling-graphs}
and plenty of practical applications~\cite{scheduling-practical,scheduling-practical2,scheduling-pracitcal3} the exact complexity of the
problem is still very far from being understood.  Since the '70s, it has
been known that the problem is $\mathsf{NP}$-hard when the number of machines is
the part of the input~\cite{ullman1975np}. 
However, the computational complexity remains unknown even when $m=3$:

\begin{problem}[\cite{GareyJ79}] \label{open:3machines}
    Is $P3 |\text{prec}, p_j=1 | C_{\max}$ solvable in polynomial time? 
\end{problem}

In fact, this is one of the four unresolved open questions from the book by Garey and Johnson~\cite{GareyJ79} and remains one of the most notorious open
question in the area (see,
e.g.,~\cite{mnich2018parameterized,levey-rothvoss,DBLP:journals/orl/BodlaenderF95}).
While papers that solve different special cases of \sched in polynomial time date back to 1961~\cite{hu1961parallel}, substantial progress on the problem was made very recently as well. In particular, a line of research initiated by Levey and Rothvo\ss~\cite{levey-rothvoss, li2021towards,garg2017quasi}, gives a quasi-polynomial approximation scheme for~$Pm |\text{prec}, p_j=1 | C_{\max}$.
In contrast to this, the exact (exponential time) complexity of the general problem has hardly been considered at all, to the best of our knowledge.
We initiate such a study in this paper.

Natural dynamic programming over subsets of the jobs solves the problem in $\Os(2^n\binom{n}{m})$ time, and an obvious question is whether this can be improved.
It is hypothesized that not all problems can be solved strictly faster than $\Os(2^n)$ (where $n$ is some natural measure of the input size): the Strong Exponential Time Hypothesis (SETH) conjectures that $k$-SAT cannot be solved in $\Os(c^n)$ time for any constant $c < 2$. Breaking the $\Os(2^n)$ barrier has been active research over the last years, with results including $\Os((2-\eps)^n)$ algorithms with $\eps > 0$ for
\textsc{Hamiltonian Cycle} in undirected graphs \cite{bjorklund2014determinant}, 
\textsc{Bin Packing} with a constant number of bins (\cite{nederlof2021faster}), and single machine scheduling with precedence constraints minimizing the total completion time \cite{cygan2014scheduling}. We show that
\sched can be added to this list of problems:

\begin{theorem} \label{thm:mainthm}
    \sched admits an $\Oh(1.995^n)$ time algorithm.
\end{theorem}

Note that~\cref{thm:mainthm} works even when the number of machines is given on
the input. In that case, decreasing the base of the exponent is the best we
can hope for with contemporary techniques.  Namely, any $2^{o(n)}$ algorithm for \sched
would result in unexpected breakthrough for \textsc{Densest $\kappa$-Subgraph}-problem
(see~\cref{sec:LB}) and a $2^{o(n)}$ time algorithm for the Biclique problem~\cite{JansenLK16} on $n$-vertex graphs. 

The starting point of our approach are two previous algorithms for $Pm |\text{prec}, p_j=1 | C_{\max}$. Recall that an (anti-)chain is a set of vertices that are pairwise (in-)comparable.
\begin{itemize}
    \item Algorithm (A): An $\Oh(n^{h(m-1)+1})$ algorithm by Dolev and Warmuth~\cite{dolev1984scheduling}, where $h$ is the maximum length of a chain (called the \emph{height}).
    \item Algorithm (B): An $\Os(\AC \cdot \binom{n}{m})$ time folklore algorithm, where $\AC$ is the number of anti chains (see Theorem~\ref{thm:DP}).
\end{itemize}

Algorithm (B) is a simple improvement of the aforementioned $\Os(2^n\binom{n}{m})$ time algorithm, where the dynamic programming table is indexed by only the elements of a subset that are maximal in the precedence order $\prec$. Algorithm (A) will be described in more detail below.

Intriguingly, Algorithm~(A) and Algorithm~(B) solve very different sets of instances quickly: A long chain cannot contribute much to the number of antichains since a chain and antichain can only overlap in one element. Optimistically, one may hope that a combination of (the ideas behind) these algorithms could make substantial progress on Open Problem~\ref{open:3machines} (by, for example, solving $P3 |\text{prec}, p_j=1 | C_{\max}$ in $2^{o(n)}$). 


In particular, a straightforward consequence of Dilworth's theorem guarantees that $\AC$ is at most $\left( 1 + \frac{n}{a} \right)^{a}$, where $a$ is the cardinality of the largest antichain (see Claim~\ref{claim:nrantichains}).
Focusing on the case when $m$ is a fixed constant, Algorithm (B) runs fast enough to achieve Theorem~\ref{thm:mainthm} whenever $a < 0.97n$. This allows us to assume that the maximum antichain is of size at least $0.97n$ and therefore there are no chains of length more than $h=0.03n$.
Unfortunately, even for constant $m$ this is still not good enough as Algorithm (A) would run in $n^{\Omega(n)}$ time.

However, the above argument gives us a stronger property: If we define $\Gc$ as the \emph{comparability graph}\footnote{The undirected graph with the jobs as vertices and edges between jobs sharing precedence constraints.} of the partial order, then in fact $\Gc$ has a vertex cover\footnote{Recall a vertex cover is a set of vertex that intersects with all edges.} of size at most $0.03n$.
%
%
%
%
Our main technical contribution is that, when we parameterize $\prec$ by the size of the vertex cover of $\Gc$ instead of by $h$, we can get a major improvement in the runtime. In particular, we get an  algorithm with a single-exponential run time and polynomial dependence on $n$ and $m$: 

\begin{theorem} \label{thm:VC}
    \sched admits $\Os(169^{|C|})$ time algorithm where $C$ is a vertex cover of the comparability graph of the precedence constraints.
\end{theorem}

Note that the fixed-parameter tractability in $|C|$ alone is not necessarily surprising or useful. To get that, one could for example guess the order in which the jobs from $C$ are processed and schedule the rest of the jobs in a greedy manner.  This, unfortunately, would yield only a $|C|^{\Oh(|C|)} \cdot
\mathrm{poly}(n)$ algorithm which is not enough to give any improvement over an exact $\Os(2^n)$ algorithm in the general setting. 

Since the runtime in Theorem~\ref{thm:VC} does not depend on the number of machines, Theorem~\ref{thm:mainthm} follows per the above discussion even when $m = \varepsilon n$ for some small constant $\varepsilon > 0$: In such cases the binomial coefficient $\binom{n}{m}$ of Algorithm (B) is still small enough to yield an $\Oh(1.995^n)$ time algorithm. For large $m > \varepsilon n$, we use a combination of the Subset Convolution technique and simple structural observations to design an $\Os(\AC+2^{n-m})$ time algorithm for \sched. See also Figure~\ref{tikz:MainThmProof}. 

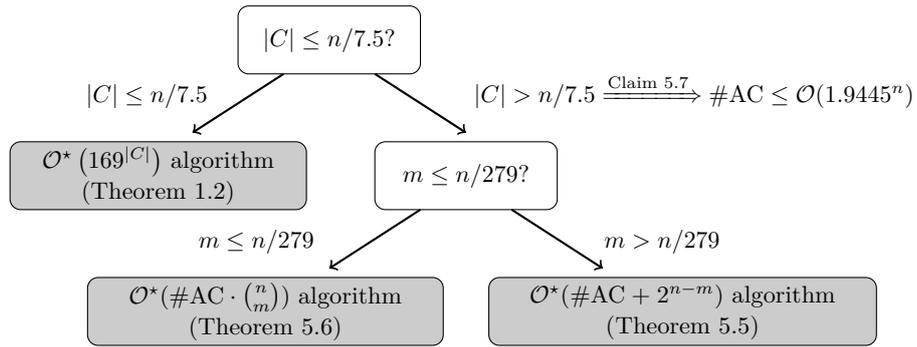
\begin{figure}[h]
\centering
     \begin{tikzpicture} [scale=0.6, every node/.style={scale=.9}] 
		\draw [rounded corners, fill = white!100!black] (0,0) rectangle (4,1.5) node[pos=.5] {$|C| \le n/7.5$? };
		
		\draw [rounded corners, fill = white!80!black] (-5,-3) rectangle (1.5,-1.5) node[pos=.5,text width=4cm, align=center] {$\Os\left(169^{|C|}\right)$ algorithm \\ (Theorem~\ref{thm:VC})};
		\draw [rounded corners, fill = white!100!black] (3,-3) rectangle (7,-1.5) node[pos=.5] {$m \le n/279$?};
		
		\draw[->,thick]  (1,0) to (-1,-1.3);
		\node at (-2,-0.5) {$|C| \le n/7.5$};
		\draw[->,thick]  (3,0) to (5,-1.3);
		\node at (10,-0.5) {$|C| > n/7.5 \xRightarrow[]{\text{Claim~\ref{claim:nrantichains}}} \AC \le \Oh(1.9445^n)$};
		
		\draw [rounded corners, fill = white!80!black] (-3.3,-6) rectangle (4.5,-4.5) node[pos=.5,text width=4cm, align=center] {$\Os(\AC\cdot \binom{n}{m})$ algorithm \\  (Theorem~\ref{thm:DP})};
	
		\draw [rounded corners, fill = white!80!black] (5.5,-6) rectangle (14,-4.5) node[pos=.5,text width=5cm, align=center ] {$\Os(\AC + 2^{n-m})$ algorithm \\ (Theorem~\ref{thm:baseline2})};
		
		\draw[->,thick]  (4,-3) to (2,-4.3);
		\node at (0.4,-3.7) {$m \le n/279$};
		\draw[->,thick]  (6,-3) to (8,-4.3);
		\node at (9.3,-3.7) {$m > n/279$};

    \end{tikzpicture}
    
    \caption{Overview of use of algorithms for proving Theorem~\ref{thm:mainthm}.}
    \label{tikz:MainThmProof}
\end{figure}

In the next paragraph, we sketch our insights behind the proof of Theorem~\ref{thm:VC}.


\paragraph*{Our approach for Theorem~\ref{thm:VC}}
The central inspiration of our algorithm is the following structural insight of the aforementioned $\Oh(n^{h(m-1)+1})$ time algorithm by Dolev and Warmuth~\cite{dolev1984scheduling}: let $z$ be the first time slot a sink (i.e., a job $v$ for which there is no precedence constraint $v \prec w$) is
scheduled. Then, there exists an optimal schedule (which is called a
\emph{zero-adjusted} schedule) for which the set of jobs before and after
timeslot $z$ can be reconstructed in polynomial time from the set of jobs scheduled \emph{at} $z$. Equipped with this observation, Dolev and
Warmuth~\cite{dolev1984scheduling} partition the schedule at timeslot $z$, (non-deterministically) guess the set of jobs scheduled at $z$ and
construct two subproblems by deducing the set of jobs scheduled before and after $z$. Then, they show that each of these subproblems
consists of a graph of height at least one less than the original graph and solve the subproblems recursively.

We extend the definition of zero-adjusted schedules and apply it to the setting of a small vertex cover. We let a \emph{sink moment} be a moment in the
schedule where at least one sink and at least one non-sink are scheduled. We define a \emph{sink-adjusted} schedule where we require that after every sink moment only successors of the jobs in the sink moment and some sinks are processed. We also show that there always exists an optimal schedule that is sink-adjusted.  

\begin{insight}
    \textbf{Key Insight}: In a sink-adjusted schedule, for each non-sink $j$ processed at time $t$ there is a chain of predecessors of $j$ intersecting all the sink moments before $t$. 
\end{insight}

\begin{figure}[ht!]
    \centering
    \includegraphics[width=0.7\textwidth]{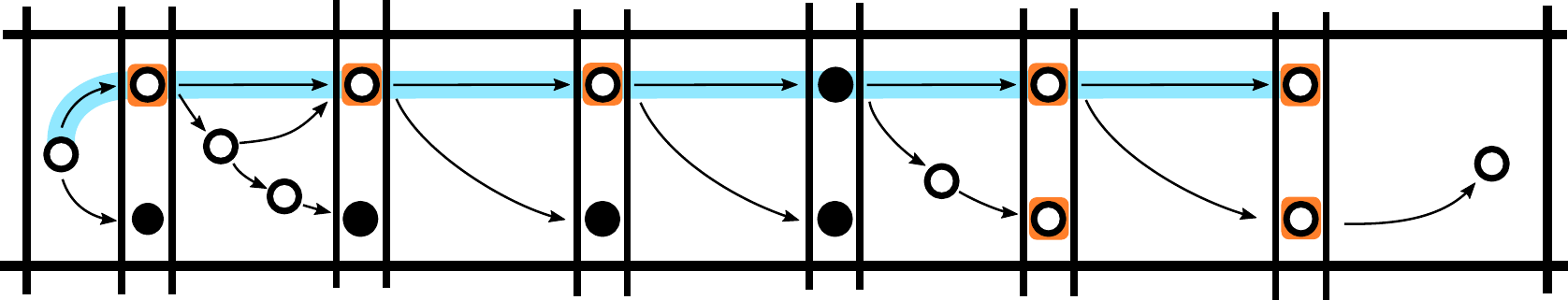}
    \caption{Illustration of Key Insight with an example of a sink-adjusted schedule. Sink-moments are distinguished between two
    bars. Jobs in the vertex cover are filled white (the remaining jobs are filled black). In a chain (highlighted blue)
    only one job is not in the vertex cover. 
    If the set jobs of the vertex cover scheduled at sink moments (the orange jobs) are known, then the position of a job that is not a sink is roughly determined, due to the Key Insight.
    }
    \label{fig:technique-VC}
\end{figure}

Note that any chain 
can contain at most one vertex not from the vertex cover. Since we are allowed to make guesses about jobs in the vertex cover, we can guess which jobs of the vertex cover are in sink moments. Subsequently, for each non-sink job $j$ we can compute the maximum length of a chain of predecessors of $j$ that are processed in sink moments, and this maximum length indicates at or in between which sink moments $j$ is scheduled (up to a small error due to the unknown existence and location of one vertex not from the vertex cover in this chain).

We split the
schedule at the moment $T'$ where roughly half of the vertex cover jobs are processed. This creates two subproblems: one formed by all jobs scheduled before $T'$ and one formed by all jobs scheduled after $T'$. Then, we use that both of these subproblems admit a
sink-adjusted schedule. For the vertex cover jobs we guess in which subproblem they are processed. We are left to partition the jobs that are not in $C$ and are sinks in the first subproblem or sources in the second subproblem (since for the remaining jobs, this is guessed or implied by the precedence constraints). 

To determine this, we find a perfect matching on a bipartite graph. One side of this graph consists of the jobs for which it is still undetermined in which subproblem they are processed. On the other side we put the possible positions for these jobs in the subproblems. Edges of this graph indicate that a job can be processed at a given position. There are no precedence constraints between these unassigned jobs since all such jobs are not in the vertex cover, and therefore a perfect matching will correspond to a feasible schedule. How to find these positions and how to define the edges of this graph is not directly clear and will be explained in Section~\ref{sec:VC}.

\subsection*{Related Work}

The \sched problem has been studied extensively from multiple angles throughout
the last decades. Ullman \cite{ullman1975np} showed that it is
$\mathsf{NP}$-complete via a reduction from $3$\textsc{-SAT}. Later, Lenstra
and Rinnooy Kan \cite{lenstra1978complexity} gave a somewhat simpler reduction
from $k$\textsc{-Clique}. 

The \sched problem is known to be solvable in polynomial time for certain structured inputs. 
Hu~\cite{hu1961parallel} gave a polynomial time algorithm when precedence
graph is a tree. This was later improved by
Sethi~\cite{Sethi:76:Scheduling-graphs} who showed that these instances can be
solved in $\Oh(n)$ time. Garey et al.~\cite{garey1983scheduling} considered a
generalization when the precedence graph is an \emph{opposing forest}, i.e., the
disjoint union of an in-forest and out-forest. They showed that the problem is $\mathsf{NP}$-hard when $m$ is
given as an input, and that the problem can be
solved in polynomial time when $m$ is a fixed constant. Papadimitriou and
Yannakakis~\cite{papadimitriou1979scheduling} gave an $\Oh(n+m)$ time algorithm
when the precedence graph is an \emph{interval order}. 
Fujii et al.~\cite{fujii1969optimal} presented the first polynomial time
algorithm when $m=2$. Later, Coffman and Graham~\cite{coffman1972optimal} gave
an alternative $\Oh(n^2)$ time algorithm for two machines. The runtime was later
improved to near-linear by Gabow~\cite{Gabow:82:An-almost-linear-algorithm} and 
finally to truly linear by Gabow and Tarjan~\cite{gabow1985linear}. 
For a more detailed overview and other variants of \sched, see the survey of Lawler
et al~\cite{lawler1993sequencing}.

\subparagraph*{Exponential Time / Parameterized Algorithms}
A natural parameter for \sched is the number of machines $m$. However, even showing that this parameterized problem is in $\mathsf{XP}$ would resolve Open
Problem~\ref{open:3machines}.  Bodlaender and
Fellows~\cite{DBLP:journals/orl/BodlaenderF95} show that problem is at least
$\mathsf{W[2]}$-hard parameterized by $m$. Recently, Bodlaender et
al.~\cite{bodlaender2021parameterized} showed that \sched parameterized by $m$
is $\mathsf{XNLP}$-hard, which implies $\mathsf{W[t]}$-hardness for every $t$.
Hence, a fixed-parameter tractable algorithm is unlikely.


Bessy and Giroudeau~\cite{bessy2019parameterized} showed that a problem called
``Scheduling Couple Tasks'' is FPT parameterized by the vertex cover of a certain associated graph.
To the best of our knowledge, this is the only other result on the parameterized complexity of scheduling when the size of the vertex cover is considered to be a parameter. 

Cygan et al.~\cite{cygan2014scheduling} study scheduling jobs of arbitrary length with precedence
constraints on one machine and proposed an $\Oh((2-\eps)^n)$ time
algorithm (for some constant $\eps > 0$). Similarly to our work, Cygan et
al.~\cite{cygan2014scheduling} consider a dynamic programming algorithm over subsets and observe that a small maximum matching in the precedence graph can be exploited to significantly reduce the number of subsets that need to be considered.

\subparagraph*{Approximation}

The \sched problem was extensively studied through the lens of approximation algorithms, where the aim is to approximate the makespan.
%
Recently, researchers analysed the problem in the important case when $m =
\Oh(1)$. In
a breakthrough paper, Levey and Rothvo\ss~\cite{levey-rothvoss} developed a
$(1+\eps)$-approximation in $\exp\left(\exp\left(\Oh(\frac{m^2}{\eps^2}
\log^2\log n)\right)\right)$ time. This was subsequently improved
by~\cite{garg2017quasi}.  The currently fastest algorithm is due to Li
\cite{li2021towards} who improved the runtime to
$n^{\Oh\left(\frac{m^4}{\varepsilon^3}\log^3\log(n)\right)}$.
Interestingly, a key step in these approaches is that approximation is easy for instances of low height. 
A prominent open question is to give a PTAS even when the number of machines is fixed (see the recent survey of Bansal~\cite{bansal2017scheduling}).

\subsection*{Organization}

In Section~\ref{sec:Prelim} we give short preliminaries. In Section~\ref{sec:Zero} we 
extend the definition of zero-adjusted schedules from Dolev and
Warmuth~\cite{dolev1984scheduling} to sink-adjusted schedules and discuss the structural insights of sink-adjusted schedules (with respect to a vertex cover).
We prove Theorem~\ref{thm:VC} in Section~\ref{sec:VC} and subsequently we show how Theorem~\ref{thm:VC} implies
Theorem~\ref{thm:mainthm} in Section~\ref{sec:exact}. We provide concluding remarks in Section~\ref{sec:conc}. Finally, we discuss (rather standard) lower bound for the problem in Appendix~\ref{sec:LB}.

\section{Preliminaries} \label{sec:Prelim}

If $B$ is a Boolean, then $\iverson{B} = 1$ if $B$ is true and $\iverson{B}=0$ if $B$ is false. We let $[N]$ denote the set of all integers $\{1,\dots,N\}$.
We use $\Ot(\cdot)$ notation to hide polylogarithmic factors and $\Os(\cdot)$ notation to hide polynomial factors in the input size.

\subparagraph*{Definitions related to the precedence constraints.}
Let the input graph $G = (V,A)$ be a precedence graph. Importantly, throughout the paper we will assume that $G$ is its transitive closure, i.e. if $(u,v)\in A$ and $(v,w) \in A$ then $(u,w)\in A$. We will interchangeably use the notations for arcs in $G$ and the partial order, i.e. $(v,w) \in A \iff v\prec w$. Similarly, we use jobs to refer to the vertices of $G$. 

The \emph{comparability graph} $\Gc = (V,E)$ of $G$ is the undirected graph obtained by replacing all directed arcs of $G$ with undirected edges. In other words, $v$ and $w$ are neighbors in $\Gc$ if and only if they are comparable to each other. 
A set $X\subseteq V$ of jobs is a \emph{chain} (\emph{antichain}) if all jobs in $X$ are pairwise comparable (incomparable).  
For a job $v$, we denote $\Pred(v) \coloneqq \{u: u \prec v\}$  as the set of all predecessors of $v$ and $\Pred[v] \coloneqq \Pred(v)\cup v$. For a set of jobs $X$, we let $\Pred(X) \coloneqq \cup_{v \in X}\Pred(v)$ and $\Pred[X] \coloneqq \cup_{v \in X}\Pred[v]$. Similarly, we define $\Succ(v) \coloneqq \{u: v \prec u \}$, $\Succ[v] \coloneqq \Succ(v) \cup \{v\}$, $\Succ(X) \coloneqq \cup_{v \in X} \Succ(v)$ and $\Succ[X] \coloneqq \cup_{v \in X} \Succ[v]$.

The \emph{height} $h(j)$ of a job $j$ is the length of the longest chain starting at job $j$, where length indicates the number of arcs in that chain. For example, the height of a job that has no successors is $0$. The height $h(G)$ of a precedence graph $G$ is equal to the maximum height of its jobs, i.e. $h(G) = \max_{j \in V} h(j)$.
We call all jobs that have no successors $\emph{sinks}$ and all jobs that have to predecessors $\emph{sources}$. For a set of jobs $X \subseteq V$ we denote $\Sinks(X)$ as all jobs of $X$ that have no successor within $X$ and $\Sources(X)$ as all jobs of $X$ that have no predecessor within $X$.

\subparagraph*{Schedules, dual graphs and dual schedules.}
A schedule $\sigma = (S_1,\dots,S_T)$ for precedence graph $G = (V,A)$ on $m$ machines is a partition of $V$ such that $|S_t|\le m$ for all $t \in [T]$ and for all $v \prec w$, if $v \in S_t$, $w\in S_{t'}$, then $t < {t'}$. We omit $G$ whenever it is clear from context. 
For a precedence graph $G = (V,A)$ we say that graph $\dual{G} = (V, \dual{A})$
is its \emph{dual} if all the arcs of $G$ are directed in the opposite
direction. We often explicitly use the fact that $\sigma = (S_1,\dots,S_T)$ is a
schedule for $G$ if and only if the dual schedule $\dual{\sigma} =
(S_T,\dots,S_1)$ is a schedule for $\dual{G}$. 

\begin{claim}\label{claim:sym}
    Let $\sigma = (S_1,\dots,S_T)$ be an optimal schedule for $G$. Then $\dual{\sigma} = (S_T,\dots,S_1)$ is an optimal schedule for $\dual{G}$.
\end{claim}

\begin{proof}
    For any jobs $u, v \in V$ with $u\prec v$, we have that $v$ is processed
    after $u$ in $\sigma$. Hence, $v$ is processed before $u$ in $\dual{\sigma}$.
    Furthermore, any time slot in $\dual{\sigma}$ contains as most $m$ jobs. Hence,
    $\dual{\sigma}$ is a feasible schedule for $\dual{G}$. Schedule $\dual{\sigma}$ is also
    optimal: if not we could reverse $\dual{\sigma}$ and find a schedule with lower
    makespan for the original instance. 
\end{proof}



\section{Sink-adjusted Schedules} \label{sec:Zero}

We define a schedule $\sigma$ as a sequence of disjoint sets of jobs
$S_1,\dots,S_T$, such that a job in set $S_i$ is processed at time slot $i$;
note that we do not need to know on which machine a job is scheduled since the
machines are identical. Naturally, if $\sigma$ is feasible then $|S_i|
\le m$ for every $i \in [T]$. The makespan of such a schedule is $T$.  For notation purposes, we use $S_{[a,b]}= \bigcup_{a \le i \le b}S_i$ to denote
the set of jobs that are processed at a time-slot between $a$ and $b$.  

Let us
stress that we do not require that all input jobs to be in $S_{[1,T]}$. In fact, in the next sections, we will apply a divide-and-conquer
technique and split the schedule into \emph{partial schedules}.  To be explicit
about this, we use $V(\sigma)$ to denote the set $S_{[1,T]}$ of jobs
assigned by $\sigma$. Naturally, a final feasible schedule needs to assign all
the input jobs. 

We prove
that we can restrict our search to schedules with certain properties, by reusing
and extending the definition of a zero-adjusted schedule from
\cite{dolev1984scheduling}. 
The definitions in the Section~\ref{sec:defs-sinks} will also be used to get an
$\Os(2^{n-m} + \AC)$ algorithm in Section~\ref{sec:exact}. Next, in
Section~\ref{sec:sinks-vc} we will consider properties of vertex cover of
sink-adjusted schedules. 
\subsection{Definition and existence of sink-adjusted schedules}
\label{sec:defs-sinks}
First, let us define the following sets for any
schedule $\sigma = (S_1\dots,S_T)$.
\begin{definition}[Sets $Z_t$ and $H_t$]
    For any time-slot $S_t$ of schedule $\sigma$,
    we define $Z_t \coloneqq S_t \cap
    \Sinks(V(\sigma))$ as all its jobs with zero height. We define set $H_t
    \coloneqq S_t \setminus \Sinks(V(\sigma))$ as all jobs in $S_t$ that have a
    height strictly greater than $0$.
\end{definition}

We then define a sink-adjusted schedule as follows (see
Figure~\ref{fig:sink-moment-def}): 

\begin{definition}[Sink-adjusted schedule and sink moments]\label{def:sinkmom}
Let $\sigma = (S_1,\ldots,S_T)$. An integer $t \in [T]$  is a \emph{sink moment}
in the schedule $\sigma$ if $0 < |H_t| < m$. 

We say that schedule $\sigma$ is \emph{sink-adjusted} if $(i)$ for
every sink moment $t \in [T]$ all jobs in $S_{t+1},\ldots,S_T$ are either
successors of some job in $S_t$, or are sinks (i.e., $S_{[t+1,T]} \subseteq
\Succ(S_t) \cup \Sinks(V(\sigma))$, and $(ii)$ all moments containing only
sinks ($S_i \subseteq \Sinks(V(\sigma))$) are scheduled after every non-sink
is scheduled.
\end{definition}

\begin{figure}[ht!]
    \centering
    \includegraphics[width=0.7\textwidth]{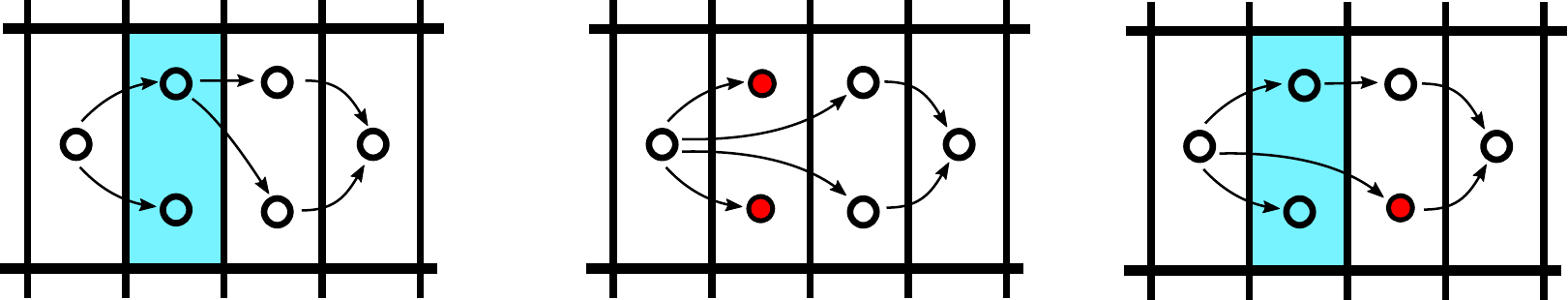}
    \caption{Examples of three schedules with makespan $4$ and $m=2$. The sink-moments are
    highlighted in blue. Only the left schedule is sink-adjusted. The middle
schedule is not sink-adjusted because the red jobs are sinks and are scheduled
before a non-sink. The right schedule is not sink-adjusted because the red job is not a
predecessor of any job in a trailing sink-moment.}
    \label{fig:sink-moment-def}
\end{figure}

Next, we prove that we can restrict our search for optimal schedules to sink-adjusted ones. The strategy behind the proof is to swap jobs until our schedule is sink-adjusted.

\begin{theorem}\label{thm:sink-adjusted}
    For every instance of \sched, there exists an optimal schedule that is sink-adjusted.
\end{theorem}
\begin{proof}
Take $\sigma = (S_1,\dots,S_T)$ to be an optimal schedule that is not sink-adjusted. 
First we prove property (ii). Let $t \in [T]$ be such that $S_t\subseteq \Sinks(V(\sigma))$, but there are non-sinks processed after $t$. Then take schedule $\sigma' = (S_1,\dots,S_{t-1},S_{t+1},\dots,S_{T},S_{t})$. In other words, we put the jobs from $S_t$ at the end of the schedule. Since $\sigma$ is optimal and the jobs in $S_t$ do not have any successors, $\sigma'$ is also optimal. So take $\sigma = \sigma'$. This step can be repeated until the second property holds.

Now we prove property (i). Let $\sigma$ be an optimal schedule that is not sink-adjusted and has the earliest $z \leq T$ such that $z$ is a sink moment, but $S_{[z+1,T]} \not\subseteq \Succ(S_z)\cup \Sinks(V(\sigma))$. 
Note that all the jobs in $\Succ(S_z)$ should be processed at or after $z+1$. Hence, there is a job $j$ with the following properties:  (1) $j \in S_i$ for some $i\in[z+1,T]$, (2) $j$ is not a sink, (3) $j\not \in \Succ(S_z)$ and (4) $\Pred(j) \cap S_{[z+1,T]} = \emptyset$. Property (4) holds because we can simply take the earliest job with properties (1-3).

Next, let us look at sink moment $z$. By definition it holds that $|H_z| < m$. This can happen because either $|Z_z| > 0$ or $|S_z| < m$. In the first case 
$|Z_z|>0$, let $j'$ be some job in $Z_z$. Observe that we can swap positions of
$j$ and $j'$ in the schedule. This new schedule is still feasible: $j'$ can be
processed later because it does not have any successors and $j$ can be processed
earlier, because it does not have any predecessors at or after time $z$.
Similarly, when $|S_z| < m$, job $j$ can be moved to empty slot in $S_z$. We can repeat this procedure until either $|S_z| = m$, $|S_z| = 0$, or $S_{[z+1,T]} \subseteq \Succ(S_z) \cup \Sinks(V(\sigma))$. Note that after this modification $\sigma$ remains an optimal schedule and none of the time slots before $z$ was changed. Because $z$ is not a sink moment anymore, the first sink moment is now after $z$. We can repeat this step until all sink moments satisfy the property (i).
\end{proof}

The reader should think about these sink moments as guidelines in the
sink-adjusted schedule that help us determine the positions of the jobs. Take
for example the first sink moment $z$.  If we know the $H_z$, then directly from
Definition~\ref{def:sinkmom} we can deduce all the jobs that are processed
before $z$ and all the jobs that are processed after $z$ (except some edge
cases, see Section~\ref{sec:VC}).  Let us remark that
the deduction of locations of
jobs based on the $H_z$ was also used by Dolev and
Warmuth~\cite{dolev1984scheduling}. 

\subsection{The structure of sink-adjusted schedules versus the vertex
cover}\label{sec:sinks-vc}
Now we assume that $C$ is a vertex cover of $\Gc[V(\sigma)]$.
We start with a simple observation about $C$:

\begin{claim} \label{claim:chain}
    Any chain in $G[V(\sigma)]$ contains at most one vertex from $V(\sigma) \setminus C$. 
\end{claim}
\begin{proof}
    For the sake of contradiction, assume that there is a chain with two different
    jobs $v,w \in V\setminus C$. These jobs are comparable to each other, hence
    there exists an edge $\{v,w\}$ in the graph $\Gc[V(\sigma)]$. However, this edge is not covered by
    $C$, which contradicts the fact that $C$ is a vertex cover of $\Gc[V(\sigma)]$. 
\end{proof}

Recall that we assumed that $G$ is equal to its transitive closure. 
We define the depth of a vertex.

\begin{definition}[Depth]
    For a set $X \subseteq V$, the \emph{depth} $\depth{X}{v}$ of a job $v \in V$
    with respect to $X$ is the length of the longest chain in $G[X \cup \{v\}]$ that
    ends in $v$.
\end{definition}

Recall, that we measure the length of a chain in its number of edges.
Note that any source has depth $0$. 
For the remainder of this section, we assume that $\sigma = (S_1,\dots,S_T)$ is a sink-adjusted schedule.
Next, we define the sinks moments of $\sigma$.

\begin{definition}[Sink Moments of the Schedule]
    Let $1 \le z(1) < \ldots < z(\ell) \le T$ be the consecutive sink moments of $\sigma$. We let
    $\SM \coloneqq \bigcup_{i \in [\ell]} S_{z(i)}$ to be the set of all jobs in
    the sink moments of $\sigma$ (we set $z(0) \coloneqq 0$ and
    $z(\ell+1) \coloneqq T+1$ for convenience).
\end{definition}

Define $\SinkVC \coloneqq C \cap \Sinks(V(\sigma))$ and let 
$\HighVC \coloneqq (C \cap \SM) \setminus \SinkVC$. 
In other words, $\SinkVC$ is the set of jobs from the vertex cover $C$ that are
sinks and $\HighVC$ is the set of jobs from $C$ that are
processed during sink moments $z(1),\ldots,z(\ell)$, but are not sinks.
Now, we show the following properties of jobs in $\HighVC$.

\begin{property}[Jobs in $\HighVC$ are almost determined]
    \label{prop:Q1}
    If $v \in \HighVC$ is scheduled at timeslot $t$ (i.e., $v \in S_t$), then it must
    be that $t = z(\depth{\HighVC}{v}+1) $ or $t = z(\depth{\HighVC}{v}+2)$. 
\end{property}
\begin{proof}
    Let us fix an arbitrary $v \in \HighVC$.  By definition of $\HighVC$, we know that $v
    \in C$, $v \notin \Sinks(\sigma(V))$ and there exists $i
    \in [\ell]$ such that $v \in S_{z(i)}$.
    Because we assumed that the schedule $\sigma$ is sink-adjusted, for any sink moment
    $z(j)$ it holds that $S_{[z(j)+1,T]} \subseteq \Succ(S_{z(j)}) \cup
    \Sinks(V(\sigma))$ where $j \in [\ell]$. 
    This implies that $t=z(\depth{\SM}{v}+1)$. Note that $\depth{\SM}{v} \geq
    \depth{\HighVC}{v}$ since
    $\HighVC \subseteq \SM$. Moreover, $\depth{\SM}{v} \leq \depth{\HighVC}{v}+1$ since any
    chain in $G[\SM]$ can contain at most one vertex in $\SM \setminus \HighVC$ since such a vertex is either a sink or not in $C$, and in the last case Claim~\ref{claim:chain} applies.
\end{proof}

\begin{property}[Jobs in $C \setminus (\HighVC \cup \SinkVC)$ are roughly determined]
    \label{prop:Q2}
    Let $v \in C \setminus (\HighVC \cup \SinkVC)$ be a vertex
    that is scheduled at moment $t \in [T]$ (i.e., $v \in S_t$),  
    then $z(\depth{\HighVC}{v}) < t < z(\depth{\HighVC}{v}+1)$ or $z(\depth{\HighVC}{v}+1) < t < z(\depth{\HighVC}{v}+2)$ 
\end{property}

\begin{proof}
The proof is similar to that of Property~\ref{prop:Q1}. Let $v \in C\setminus (\HighVC \cup \SinkVC)$. Hence, $v \in C$, $v \notin \Sinks(\sigma(V))$ and $v$ is not processed at any sink moment. 
    Because we assumed that the schedule $\sigma$ is sink-adjusted, for any sink moment
    $z(j)$ it holds that $S_{[z(j)+1,T]} \subseteq \Succ(S_{z(j)}) \cup
    \Sinks(V(\sigma))$ where $j \in [\ell]$. 
    This implies that $z(\depth{\SM}{v}) < t < z(\depth{\SM}{v}+1)$. Note that $\depth{\SM}{v} \geq
    \depth{\HighVC}{v}$ since
    $\HighVC \subseteq \SM$. Moreover, $\depth{\SM}{v} \leq \depth{\HighVC}{v}+1$ since any
    chain in $G[\SM]$ can contain at most one vertex in $\SM \setminus \HighVC$ since such a vertex is either a sink or not in $C$, and in the last case Claim~\ref{claim:chain} applies.
    Note that $v \in \HighVC \cup \SinkVC$, so $v$ cannot be processed at any sink moment. Hence the boundaries on $t$ follow.
\end{proof}

\begin{figure}[ht!]
    \centering
    \includegraphics[width=\textwidth]{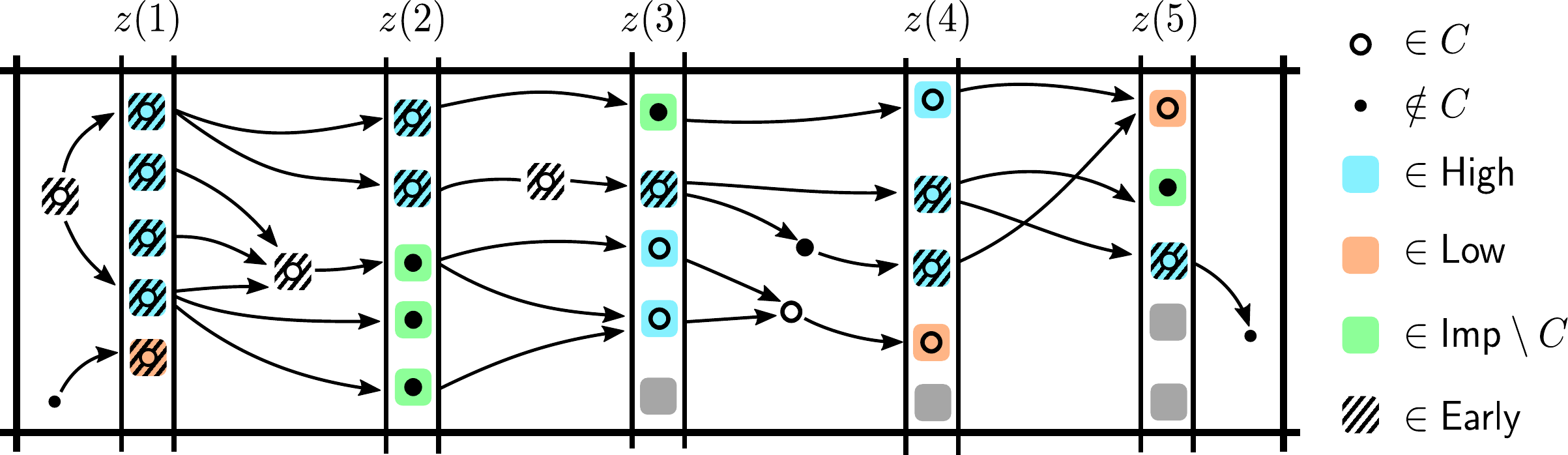}
    \caption{Figure presents a schedule with five sink moments $z(1),\ldots,z(5)$. Jobs from Vertex
    Cover are marked with a black ring (the remaining jobs are depicted with a black-filled circle).
    Set $\SM$ contains every job in the sink-moment (highlighted either blue, green or red). Jobs from $\HighVC$ are highlighted in blue. Jobs from $\SinkVC$ are highlighted red. Jobs from $\SM \setminus C$ are highlighted
    green. Early jobs are wrapped with a solid black border. Gray rectangles are
    empty slots in the schedule. Observe that for a fixed chain there exists a
    single moment after which every job in it is late. This moment corresponds to
    the job from $\SM \setminus C$ on this chain.}
    \label{fig:definitions}
\end{figure}

Next, we define Early Jobs. See Figure~\ref{fig:definitions} for example of
$\HighVC$, $\SinkVC$ and intuition behind $\EarlyVC$ jobs.

\begin{definition}[Early Jobs]
    \label{def:early-job}
    We say a job $v \in S_t$ is \emph{early} if either
    \[
    \begin{aligned}
        v &\in \HighVC \cup \SinkVC & \textnormal{and } &  &t&=z(\depth{\HighVC}{v}+1), &\textnormal{ or}\\
        v &\in C \setminus  (\HighVC \cup \SinkVC) & \textnormal{and } &
        z(\depth{\HighVC}{v})) < &t & <  z(\depth{\HighVC}{v}+1). &
    \end{aligned}
    \]
\end{definition}
If a job is not early, we call it \emph{late}. By Property~\ref{prop:Q1}, a late
job $v$ in $\HighVC$ is scheduled at $z(\depth{\HighVC}{v}+2)$. By Property~\ref{prop:Q2} a late job $v$ in $C \setminus (\HighVC \cup \SinkVC)$ is scheduled in between
$z(\depth{\HighVC }{v}+1)$ and $z(\depth{\HighVC }{v}+2)$. Additionally it will be useful in Section~\ref{sec:VC} to know which jobs in $\SinkVC$ are early and late in order to ensure that precedence constraints $v \prec w$ with $w \in \SinkVC$ and $v$ are not scheduled at the same sink moment.

Crucially, if we guess the set $\HighVC$ of a sink-adjusted schedule $\sigma$, and which non-sink jobs are early and which non-sink jobs are late we can already deduce for each job in 
\[
 \HighVC \cup (C \setminus (\HighVC \cup \SinkVC)) \cup (V(\sigma) \setminus (C \cup \Sinks(V(\sigma))) = V(\sigma) \setminus \Sinks(V(\sigma))
\]
on (or in between) which sink-moment it is scheduled.

\begin{property}[Jobs in $V(\sigma) \setminus (C \cup \Sinks(V(\sigma)))$ are also roughly determined]\label{prop:Q3}
    Let $v \in V(\sigma) \setminus (C \cup \Sinks(V(\sigma)))$ be a vertex that is scheduled at moment $t \in [T]$ (i.e., $v \in S_t$).
    Then $z(\depth{\HighVC }{v}) < t \le z(\depth{\HighVC }{v}+1)$.
\end{property}
\begin{proof}
    The proof is similar to that of Property~\ref{prop:Q1}. Let $v \in V(\sigma)\setminus(C \cup \Sinks(V(\sigma))$. Hence, $v \not\in C$, $v \notin \Sinks(V(\sigma))$. Note that $v$ might or might not be scheduled at a sink moment. 
    Because we assumed that the schedule $\sigma$ is sink-adjusted, for any sink moment
    $z(j)$ it holds that $S_{[z(j)+1,T]} \subseteq \Succ(S_{z(j)}) \cup
    \Sinks(V(\sigma))$ where $j \in [\ell]$. 
    This implies that $z(\depth{\SM}{v}) < t \le z(\depth{\SM}{v}+1)$. Note that $\depth{\SM}{v} \geq
    \depth{\HighVC}{v}$ since
    $\HighVC \subseteq \SM$. Moreover, $\depth{\SM}{v} \leq \depth{\HighVC}{v}$ because $v \not \in C$ and thus $v$ is the only element in a chain ending in $v$ that is not in $C$ by Claim~\ref{claim:chain}.
\end{proof}

\section{Single Exponential FPT Algorithm when Parameterized by Vertex Cover of
the Comparability Graph} \label{sec:VC}
In this section we prove Theorem~\ref{thm:VC} and give an $\Os(169^{|C|})$ time algorithm
for \sched. We assume that the vertex cover $C \subseteq V$ of the comparability
graph is given as input (if not, we can easily find it with the standard algorithm in $\Os(2^{|C|})$ time).
Also, we assume that the deadline is $T$ and that there are exactly $m\cdot T$
jobs to be processed; this can be ensured by adding $m\cdot T - n$ jobs without any precedence constraints. Note
that this operation does not increase the size of the vertex cover of $\Gc$, as
no edge is added to the precedence graph. Moreover, the number of added jobs is
bounded by $n \cdot m \le n^2$, which is only an additional polynomial factor in
the running time.  For convenience we use the following notation throughout this section:
\begin{definition}
We call $(S_1,\ldots,S_T)$ a \emph{tight $m$-schedule for $G$} if the $S_i$'s partition $V(G)$ and for all $i\in[T]$ we have $|S_i|=m$.
\end{definition}
If $G$ is clear from the context, it will be omitted.
By the above discussion, we can restrict attention to detecting tight $m$-schedules.

\subsection{Middle-adjusting schedules and their fingerprints}

We will split the schedule at some time slot $T'$ into two subproblems and solve
them recursively. The issue with this approach is that even if we know which
jobs are scheduled at time slot $T'$ we still need to determine which jobs are
scheduled before and after $T'$. To assist us with this task, we restrict our
search to schedules with a specific structure. We call these structure
\emph{middle-adjusted} schedule.

\begin{definition}[Middle-adjusted Schedule]
    We say that a schedule $\sigma = (S_1,\ldots,S_T)$ is \emph{middle-adjusted}
    at timeslot $T'$ if $\sigma_L \coloneqq (S_1,\ldots,S_{T'-1})$ and
    $\dual{\sigma_R} \coloneqq (S_T,\ldots,S_{T'+1})$ are both sink-adjusted.
\end{definition}

\begin{lemma} \label{lemma:composition}
    For any tight $m$-schedule $\sigma=(S_1,\ldots,S_{T})$ and time $T' \in
    [T]$, there is a tight $m$-schedule $\sigma' = (S'_1,\ldots,S'_T)$
    middle-adjusted at timeslot $T'$ such that $S'_{T'}=S'_{T'}$,
    $S'_{[1,T'-1]}=S_{[1,T'-1]}$ and $S'_{[T'+1,T]}=S_{[T'+1,T]}$.
\end{lemma}
\begin{proof}
    Let $\sigma_L:= (S_1,\dots,S_{T'-1})$ and $\dual{\sigma_R}:= (S_{T },\dots,S_{T'+1})$. 
    By Theorem~\ref{thm:sink-adjusted}, there are tight $m$-schedules of the instances $G[V(\sigma_L)]$ and $G[V(\sigma_R)]$ (with precedence constraints reversed) of \sched that are sink-adjusted. Concatenating these schedules with $S_{T'}$ in between results in a middle-adjusted schedule.
\end{proof}

Our goal is to deduce the set of jobs processed at $\sigma_L$ and $\sigma_R$
based on the fact that our schedule is middle-adjusted and properties of the
vertices of $C$ with respect to the schedule. Since $C$ is small, we can guess
these properties with few guesses.  The aforementioned properties are formalized
in the following definition:

\begin{definition}[Fingerprint]
Let $\sigma=(\sigma_L,S_{T'},\sigma_R)$ be middle-adjusted schedule at $T'$. Let
\[
\begin{aligned}
C_L &\coloneqq V(\sigma_L) \cap C, \\
C_R &\coloneqq V(\sigma_R) \cap C, \\
\SinkVC_L &\coloneqq \Sinks(V(\sigma_L)) \cap C_L,\\
\SinkVC_R &\coloneqq \Sinks(V(\dual{\sigma_R})) \cap C_R,\\
\HighVC_L &\coloneqq \{ v \in C_L \setminus \Sinks(V(\sigma_L)): v \textnormal{ scheduled at a sink moment of $\sigma_L$}  \}, \\
\HighVC_R &\coloneqq \{ v \in C_R \setminus \Sinks(V(\dual{\sigma_R})): v \textnormal{ scheduled at a sink moment of $\dual{\sigma_R}$}  \}, \\  
\EarlyVC_L &\coloneqq \{ v \in C_L : v \textnormal{ is early in $\sigma_L$} \}, \\  
\EarlyVC_R &\coloneqq \{ v \in C_R : v \textnormal{ is early in $\dual{\sigma_R}$} \}. 
\end{aligned}
\]
We call $8$-tuple
$(C_L,C_R,\SinkVC_L,\SinkVC_R,\HighVC_L,\HighVC_R,\EarlyVC_L,\EarlyVC_R)$ the \emph{fingerprint} of $\sigma$.
\end{definition}

The following will be useful to bound the runtime of our algorithm and is easy to check by case analysis:
\begin{claim}
    There are at most $13^{|C|}$ different fingerprints.
\end{claim}
\begin{proof}
    Let $e \in C$. If $e \in C_M$ it cannot be in any of the other sets. If $e
    \in C_L$, it can be in $\HighVC_L$ and $\SinkVC_L$, but not in both.
    Additionally, independently it could be in $\EarlyVC_L$. Thus, there are
    $3\cdot2=6$ possibilities (see $C_L$ cell in Figure~\ref{fig:composition}).
    Similarly, there are $6$ possibilities if $e \in C_R$. Thus in total there
    are $1+6+6=13$ possibilities per element in $C$. 
\end{proof}

\usetikzlibrary{patterns.meta}
\usetikzlibrary{patterns}

\usetikzlibrary{decorations.pathreplacing,angles,quotes}
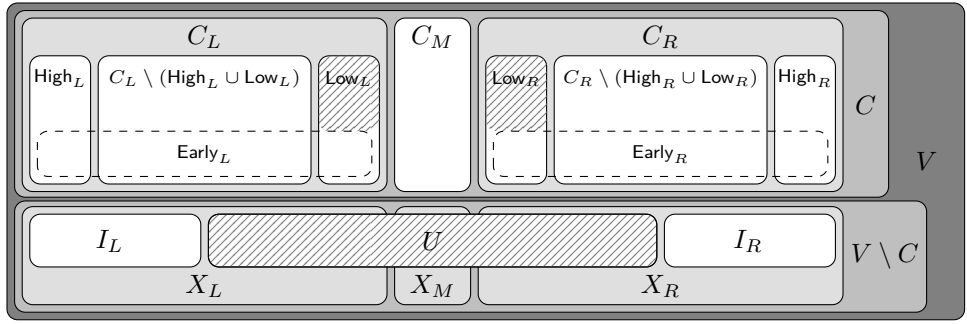
\begin{figure}
\centering
 \begin{tikzpicture} [scale = 1] 
    \draw[rounded corners, fill = gray] (12.5,2.1) rectangle (-.1,-2.1);
    \node at (12.,0) {$V$};
    
    \draw[rounded corners, fill=gray!50!white] (0,-.475) rectangle (11.5,2);
    \node at (11.2,.75) {$C$};

    \draw[rounded corners, fill = gray!25!white] (0.1,-.4) rectangle (4.9,1.9);
    \node at (2.5,1.65) {$C_L$};
    
    \draw[rounded corners, fill = white] (0.2,-.3) rectangle (1.0,1.4);
    \node at (.6,1.1) {\scriptsize $\HighVC_L$};
    
    \draw[rounded corners, fill = white] (1.1,-.3) rectangle (3.9,1.4);
    \node at (2.5,1.1) {\scriptsize $C_L\setminus (\HighVC_L \cup \SinkVC_L)$};
    
    \draw[rounded corners, fill = white] (4.0,-.3) rectangle (4.8,1.4);
    \fill[rounded corners, pattern=north east lines, pattern color =gray] (4,.4) rectangle (4.8,1.4); 
    \node at (4.4,1.1) {\scriptsize $\SinkVC_L$};

    \draw[rounded corners, fill = white] (5,-.4) rectangle (6,1.9);
    \node at (5.5,1.65) {$C_M$};

    \draw[rounded corners, fill = gray!25!white] (6.1,-.4) rectangle (10.9,1.9);
    \node at (8.5,1.65) {$C_R$};

    \draw[rounded corners, fill = white] (6.2,-.3) rectangle (7.0,1.4);
     \fill[rounded corners, pattern=north east lines, pattern color =gray] (6.2,.4) rectangle (7.0,1.4); 
    \node at (6.6,1.1) {\scriptsize $\SinkVC_R$};
    
    \draw[rounded corners, fill = white] (7.1,-.3) rectangle (9.9,1.4);
    \node at (8.5,1.1) {\scriptsize $C_R\setminus (\HighVC_R \cup \SinkVC_R)$};
    
    \draw[rounded corners, fill = white] (10.0,-.3) rectangle (10.8,1.4);
    \node at (10.4,1.1) {\scriptsize $\HighVC_R$};

    \draw[rounded corners, fill = gray!50!white] (0,-.525) rectangle (12,-2);
    \node at (11.45,-1.25) {$V\setminus C$};
    
    \draw[rounded corners, fill = gray!25!white] (0.1,-0.6) rectangle (4.9,-1.9);
    \node at (2.5,-1.65) {$X_L$};
    
    \draw[rounded corners, fill = white] (0.2,-0.7) rectangle (2.45,-1.4);
    \node at (1.25,-1.05) {$I_L$};

    \draw[rounded corners, fill = gray!25!white] (5,-0.6) rectangle (6,-1.9);
    \node at (5.5,-1.65) {$X_M$};
        
    \draw[rounded corners, fill = gray!25!white] (6.1,-0.6) rectangle (10.9,-1.9);
    \node at (8.5,-1.65) {$X_R$};

    \draw[rounded corners, fill = white] (8.55,-0.7) rectangle (10.8,-1.4);
    \node at (9.65,-1.05) {$I_R$};
    
    \draw[rounded corners, fill = white] (2.55,-0.7) rectangle (8.45,-1.4);
    \draw[rounded corners, pattern=north east lines, pattern color =gray] (2.55,-0.7) rectangle (8.45,-1.4); 
    \node at (5.5,-1.05) {$U$};
    
    \draw[rounded corners, dashed] (0.3,-.2) rectangle (4.7,.4);
    \node at (2.5,.1) {\scriptsize $\EarlyVC_L$};
    
    \draw[rounded corners, dashed] (6.3,-.2) rectangle (10.7,.4);
    \node at (8.5,.1) {\scriptsize $\EarlyVC_R$};

    \end{tikzpicture}   
    \caption{Venn diagram of the sets often used in Section~\ref{sec:VC}. Recall that by definition $I_L = \Pred(C_L)\setminus C_L$ and $I_R = \Succ(C_R)\setminus C_R$. The dashed area is equal to $U'$, defined in Subsection~\ref{subsec:divide}.}
    \label{fig:composition}
\end{figure}

\subsection{The algorithm}

An overview of the algorithm is described in Algorithm~\ref{algorithm1}. It is
given a precedence graph $G$, number of machines $m$, and a vertex cover $C$ of
$\Gc$ as input. The Algorithm outputs a tight $m$-schedule if it exists, and
``False'' otherwise.

\let\oldnl\nl
\newcommand{\nonl}{\renewcommand{\nl}{\let\nl\oldnl}}
\begin{algorithm}[h!] 
\nonl\noindent\textbf{Algorithm} $\mathtt{schedule}(G,C,m)$ \\ 
\SetAlgoLined
\ForEach{$T' \in [1,T]$}{
\label{line:fingerprint}\ForEach{\textnormal{fingerprint}
$f=(C_L,C_R,\SinkVC_L,\SinkVC_R,\HighVC_L,\HighVC_R,\EarlyVC_L,\EarlyVC_R)$}{
\label{line:vccheck}\If{$|C_L|, |C_R| \le |C|/2$}{
$(X_L,X_M,X_R) \gets \mathtt{divide}(G,m,T',C,f)$\\
$\sigma_L \gets \mathtt{schedule}(G[C_L \cup X_L],C_L,m)$\\
$\sigma_R \gets \mathtt{schedule}(G[C_R \cup X_R],C_R,m)$\\
\If{$\sigma = (\sigma_L, C_M \cup X_M, \sigma_R)$ \textnormal{is a tight $m$-schedule for $G$}}{\label{line:check}
\Return $\sigma$\label{line:output}
}
}
}
}
\Return False
\caption{Algorithm for Theorem~\ref{thm:VC}.}
\label{algorithm1}
\end{algorithm}

The first step of the algorithm is to guess integer $T' \in [T]$ such that at
most half of the jobs from $C$ are processed before $T'$ and at most half of the jobs
from $C$ are processed after $T'$. 
Subsequently, we guess the fingerprint $f$ of a middle-adjusted schedule $(\sigma_L,S_{T'},\sigma_R)$
Effectively, we guess for every job in $C$ whether it is processed in
$\sigma_L$, at $T'$ or in $\sigma_R$, and whether it is in $\SinkVC$, $\HighVC$ and $\EarlyVC$.

If we have guessed correctly, then we can deduce that jobs $\Pred(C_L)\setminus
C_L$ must be in $\sigma_L$ and the jobs in $\Succ(C_R)\setminus C_R$ are in
$\sigma_R$.  We are not done yet, as the position of the remaining jobs from $V
\setminus C$ is still not known. To solve this, we employ a subroutine
$\mathtt{divide}$ that tells us for all jobs in $V \setminus C$ whether they are
scheduled in $\sigma_L$, at $T'$ or $\sigma_R$, by making use of the
fingerprint. Formally: 

\begin{lemma}
    \label{claim:reconstruction}
    There is a polynomial time algorithm $\mathtt{divide}$ that, given as input
    precedence graph $G$, integers $m,T' \in \mathbb{N}$, vertex cover $C$ of $\Gc$, and a fingerprint
    \[
            f=(C_L,C_R,\SinkVC_L,\SinkVC_R,\HighVC_L,\HighVC_R,\EarlyVC_L,\EarlyVC_R),
    \] 
    finds a partition $X_L,X_M,X_R$ of $V \setminus C$ with the following property:
    If $f$ is the fingerprint of a tight $m$-schedule $\sigma$ of $G$ that is middle-adjusted at time $T'$, then $G[C_L \cup X_L]$ and $G[C_R \cup X_R]$ have tight $m$-schedules, $|X_M \cup (C \setminus (C_L\cup C_R))| = m$, $\Pred(C \setminus C_R) \subseteq C_L \cup X_L$, and $\Succ(C \setminus C_L) \subseteq C_R \cup X_R$.
\end{lemma}
This lemma will be proved in the next subsection.

With the partition of $V \setminus C$ into $X_L,X_M,X_R$ in hand, we can solve
the associated two subproblems with substantially smaller vertex covers $C_L$
and $C_R$ recursively. If the combination results in a tight $m$-schedule we
return it, and if such a schedule is never found we return ``False'' .  This
concludes the description of the algorithm, except the description of the
subroutine $\mathtt{divide}$.

\subparagraph*{Run time analysis.}
There are $13^{|C|}$ guesses for fingerprint $f$ in Algorithm~\ref{algorithm1}.
Additionally, there are at most $n$ possible guesses of $T'$. After all
guesses are successful, then in polynomial time we determine the set of jobs in
$X_L$, $X_M$ and $X_R$ by~\cref{claim:reconstruction} and with that, the jobs for the two subproblems: $C_L\cup X_L$ and $C_R\cup X_R$. Subsequently, we
recurse, and solve these two instances of \sched: one with jobs $C_L\cup X_L$ and one with
$C_R\cup X_R$. Observe that by definition $C_L$ is a vertex cover of $C_L\cup X_L$ and $C_R$ is a vertex cover of $C_R\cup X_R$. Moreover $|C_L|,|C_R| \le |C|/2$. Therefore, the
total runtime $T(|C|)$ of the algorithm is bounded by:
\begin{displaymath}
    T(|C|) \le 13^{|C|} \cdot T\left(\frac{|C|}{2}\right) \cdot  n^{\Oh(1)}.
\end{displaymath}
Therefore, the total runtime of the algorithm is $T(|C|) \le \Os(169^{|C|})$ as claimed.

\subparagraph{Correctness.}
We claim that $\mathtt{schedule}(G,C,m)$ returns a tight $m$-schedule if it exists, and that it returns ``False'' otherwise. Note that Algorithm~\ref{algorithm1} checks for feasibility in Line~\ref{line:check}, so if  there is no tight $m$-schedule it will always return ``False''.

Thus, let us focus on the first part.
Let $(S_1,\ldots,S_T)$ be a tight $m$-schedule.
Let $T'$ be the smallest integer such that $|S_{[1,T']} \cap C| \geq |C|/2$.
Then by Lemma~\ref{lemma:composition}, there is a tight $m$-schedule $\sigma=
(\sigma_L,S_{T'},\sigma_R)$ that is middle adjusted at time $T'$ such that
$V(\sigma_L)=S_{[1,T'-1]}$.
Consider the iteration of the loop at Line~\ref{line:fingerprint} where we pick the fingerprint of $\sigma$. 
By the choice $T'$ we have that $|V(\sigma_L) \cap C|, |V(\sigma_R) \cap C| \leq |C|/2$, and hence the check at Line~\ref{line:vccheck} is verified. Let $C_M = C \setminus (C_L \cup C_R)$.
By Lemma~\ref{claim:reconstruction}, we find $X_L,X_M,X_R$ such that $|X_M \cup C_M| = m$, there is a tight $m$-schedule $\sigma_L'$ for $G[C_L \cup X_L]$ and a tight $m$-schedule $\sigma_R'$ for $G[C_R \cup X_R]$. We claim that $\sigma'=(\sigma_L',(X_M\cup C_M),\sigma_R')$ is a tight $m$-schedule and hence it will be output at Line~\ref{line:output}. To see this, note we only need to check whether precedence constraints between vertices from different parts of the partition $V(\sigma_L'), X_M\cup C_M, V(\sigma_R')$ are satisfied. Let $v \preceq w$ be such a constraint. Note that either $v \in C$ or $w \in C$ (or both), since $C$ is a vertex cover of $\Gc$. If $v \in C$ then the constraint $v \prec w$ is satisfied since $v \in C_L$ or $\Succ(v) \subseteq V(\sigma_R)$ by Lemma~\ref{claim:reconstruction}. Similarly, if $w \in C$ then the constraint $v \prec w$ is satisfied since $w \in C_R $ or $\Pred(v) \subseteq V(\sigma_L)$.
Thus $\sigma'$ is a tight $m$-schedule and the correctness follows.

\subsection{Dividing the jobs: The proof of  Lemma~\ref{claim:reconstruction} } \label{subsec:divide}

In this subsection we prove Lemma~\ref{claim:reconstruction}. Let us assume
that $f$ is a fingerprint of a middle-adjusted schedule $\sigma \coloneqq
(\sigma_L,S_{T'},\sigma_R)$ at $T'$ (hence $\sigma_L$ and $\dual{\sigma_R}$ are
both sink-adjusted).

First of all, we can deduce that jobs in $I_L \coloneqq \Pred(C_L)\setminus
C_L$ must be processed in $\sigma_L$ because their successors are in
$\sigma_L$. Similarly every job in $I_R \coloneqq \Succ(C_R)\setminus C_R$
needs to be in $\sigma_R$. It remains to assign jobs in 
$U \coloneqq V \setminus (C \cup I_R \cup I_L)$. For this, we will actually assign jobs from $U'$ using a perfect matching on a bipartite graph, where
\[ U' \coloneqq U \cup (\SinkVC_L\setminus \EarlyVC_L) \cup (\SinkVC_R \setminus \EarlyVC_R).
\]
We show that for the jobs that are not in $U'$, we know roughly where they are using the fingerprint and Properties~\ref{prop:Q1}, \ref{prop:Q2} and \ref{prop:Q3} for schedules $\sigma_L$ and $\dual{\sigma_R}$. 

We will determine where the jobs from $U'$ go using a perfect
matching on a bipartite graph $H = ((U',P),F)$. The set $P\subset 
[T]\times[m]$ consists of positions at which the jobs of $U'$ are processed in $\sigma$ and an edge 
$(u,(t,j)) \in F$ will indicate that $u\in U'$ can be processed at time $t \in [T]$. The `$j$' indicates that it is the $j$th machine that will process the job. 

We will claim later that we can independently determine for each job in $U'$
whether it can be processed at a specific position in $P$. As such, finding
a perfect matching of graph $H$ will determine the position of each job in
$U'$. Note that jobs in $U'$ need not be assigned at their positions in $\sigma$ with this method, but they will be assigned at a position that will make an $m$-tight schedule.

\subparagraph*{Construction of $P$.}
To construct this bipartite graph, we first find the set of possible
positions $P$ where jobs from $U'$ are processed.  At $T'$ the jobs from
$C_M$ are processed, so there are $m-|C_M|$ jobs from $U'$ processed there.
We add positions $(T',j)$ for $j\in [m-|C_M|]$ to $P$.
    
Let us now define the positions in $P$ for $t<T'$, i.e. the positions in
$\sigma_L$. Let $z_L$ be the first timeslot in $\sigma_L$ at which only sinks
are processed. Since all jobs from $U'$ are sinks in $\sigma_L$, they can only
be processed at a sink moment of $\sigma_L$ or at or after $z_L$. Hence, to find
the correct positions, we need the value of $z_L$ and the number of jobs from
$U'$ at each sink moment of $\sigma_L$.  For this we first define blocks:
\begin{definition}
Let $z(1),\dots,z(\ell)$ be the sink moments of $\sigma_L$.
Then for $i \in
[1,\ell]$ we define the \emph{$i$th block} $B_i \coloneqq  [z(i-1)+1,z(i)]$ and we let $B_{\ell+1} = [z(\ell)+1,T'-1]$. Recall that $z(0) = 0$.
The \emph{length} of a block $[l,r]$ is defined as $r-l+1$ (i.e., the length of the interval).
\end{definition}
 We will show that for many jobs, we can determine in which block they are processed.

\begin{claim} \label{claim:blocks}
Let $\sigma$ be a middle-adjusted tight $m$-schedule. Given as input the precedence graph $G$, integer $m$, and fingerprint $f$ of $\sigma$ we can determine in polynomial time:
\begin{enumerate}[(1)]
    \item\label{item1} for $v \in \HighVC_L \cup (\SinkVC_L \cap \EarlyVC_L)$ at which time they are processed, and
    \item\label{item2} for $v \in \Pred[C_L]\setminus(\SinkVC_L \setminus \EarlyVC_L)$ at which block they are processed,
    \item\label{item3} the length of each block,
    \item\label{item4} the value of $z_L$.
\end{enumerate}
\end{claim}
\begin{proof}
For each job in $\HighVC_L$ we know whether it is early or
late, so using Property~\ref{prop:Q1} we know the exact sink moment it is
processed, and as a consequence also in which block. 
For a job in $\SinkVC_L \cap \EarlyVC_L$, we know by
Definition~\ref{def:early-job} at which sink moment it is processed and as a
consequence also in which block. Thus, to establish Item~(\ref{item1}) we only
need to determine when all sink moments are exactly (or in other
words the length of each block).

For jobs in $C_L
\setminus (\HighVC_L \cup \SinkVC_L)$, we know whether it is early or late and we
use Property~\ref{prop:Q2} to find in which block it is processed. 
Recall that $I_L \coloneqq \Pred(C_L)\setminus C_L$, so $I_L \subseteq V(\sigma_L) \setminus (C_L \cup \Sinks(V(\sigma_L))$ as all
    jobs in $I_L$ are not in $C_L$ and they have some successor in $C_L$. Hence
    for any job in $I_L$, Property~\ref{prop:Q3} tells us exactly in which block
    it is processed. This
concludes the proof of Item~(\ref{item2}). 

Note that, by Item~(\ref{item2}), all jobs from $V(\sigma_L)$ for which we
have not determined the block in which they are processed yet are all sinks
in $\sigma_L$. Recall that $z_L$ is the first time slot such that $S_{z_L}
\subseteq \Sinks(V(\sigma_L))$. Hence sinks from $\sigma_L$ can only be processed at sink moments of $\sigma_L$ or after or at $z_L$.
Therefore, for each block $B_i$ with $i \le \ell$ the only jobs that have
not been assigned to it are at the sink moment $z(i)$. Hence, we can determine
the length of each block as follows: If $n_i$ is
the number of jobs from $\Pred[C_L]\setminus(\SinkVC_L \setminus \EarlyVC_L)$ in
block $i$, then the length of block $i$ must be $\lceil n_i/m \rceil$. As a
consequence, we do not only know at which sink moment the jobs from $\HighVC_L$
and $\SinkVC\cap \EarlyVC$ are processed, but also at which time. This
established Item~(\ref{item1}) and Item~(\ref{item3}).

Finally, for Item~(\ref{item4}), we can compute the value $z_L$ by computing how many jobs from $\Pred[C_L]\setminus(\SinkVC_L \setminus \EarlyVC_L)$ are processed in the $(\ell+1)$th block; if the number of such jobs is $x$ then $z_L$ will be equal to $z(\ell) +\lceil x / m \rceil$, by the same reasoning as above.
\end{proof}

We need to decide for each vertex in 
$U'$ whether it is scheduled in $\sigma_L$,  at $T'$ or in $\sigma_R$. Note that the set $U' \cap V(\sigma_L)$ is equal to the set of jobs for which we do not know  by Claim~\ref{claim:blocks} at which block they are processed. 
As a consequence, if $u \in U'$ is processed in $\sigma_L$, then it is a sink and it can only be processed at a sink moment or after or at $z_L$. 

Let $z(i)$ be a sink moment of $\sigma_L$, we will describe how to compute $|S_{z(i)}\cap U'|$, i.e. the number of positions that we need to create in the bipartite graph for time $z(i)$. Claim~\ref{claim:blocks} gives the number of non-$U'$ jobs within that block, say $n_i$. Hence, the number of positions at $z(i)$ for jobs from $U'$ is equal to $(m - n_i )\bmod m$.
Therefore, we add $(z(i),j)$ to $P$ for all $j \in [(m - n_i )\bmod m]$.

For all $t \in [z_L,T'-1]$ we create positions $(t,j)$ for every $j\in[m]$;
each of these moments only contains sinks of $\sigma_L$. Note that all jobs
processed at or after $z_L$ are jobs from $U'$, as any job in $\SinkVC_L\cap
\EarlyVC_L$ is processed at some sink moment by definition of $\EarlyVC$. \\
    
For the positions $t > T'$ in $P$, we can use the same strategy. Note that
by symmetry Claim~\ref{claim:blocks} holds also for $\dual{\sigma_R}$. This way,
we can find all possible positions for jobs of $U'$ in a middle-adjusted
schedule $\sigma$ in polynomial time, given $m$, the input graph and the
fingerprint $f$ of $\sigma$. 

\subparagraph*{Construction of edges $F$.}   
To define the edges of the bipartite graph $H$ and prove that any perfect matching on this bipartite graph relates to a feasible schedule, we will use the following claim.

\begin{claim} \label{claim:intervals}
Given $T'$, the fingerprint $f$ and precedence graph $G$, we can determine in polynomial time for each $v \in U'$ an interval $[l_v,r_v]$ such that 
\begin{enumerate}[(1)]
    \item $\Pred(v) \setminus U'$ is scheduled before $l_v$,
    \item $\Succ(v) \setminus U'$ is scheduled after $r_v$,
    \item  $v$ is scheduled in interval $[l_v,r_v]$ in $\sigma$.
\end{enumerate}
Furthermore if $u,v \in U'$ and $u \prec v$ then $r_u < l_v$. Finally, if $u \in U'$, $v \in C_M$ and $u\prec v$ then $r_u < T'$ and similarly if $u \in U'$, $v \in C_M$ and $v\prec u$ then $T' < l_u$.
\end{claim}
\begin{proof}

Recall that $U'$ is the union of $U$, $(\SinkVC_L\setminus \EarlyVC_L)$, and
$(\SinkVC_R \setminus \EarlyVC_R)$. We will prove the claim for each of these
three sets separately. The cases $v \in (\SinkVC_L \setminus \EarlyVC_L )$ and $v \in
(\SinkVC_R\setminus \EarlyVC_R)$ are symmetric and we consider them first. 

As before, let $\ell$ be the number of sink moments in $\sigma_L$ and $z(i)$ the time of the $i$th sink moment of $\sigma_L$.  
Let $k$ be the number of sink
moments in $\dual{\sigma_R}$ and $y(i)$ the time of the $i$th sink moment of
$\dual{\sigma_R}$ in $\sigma$. Let $z_R \in [T'+1, T]$ be the first moment of
$\sigma$ where a non-source of $V(\sigma_R)$ is processed (see
Figure~\ref{fig:zlzr} for schematic definition of positions $z(i),y(i)$ and
$z_L$ and $z_R$).
Define $z'(i)$ as $z(i)$ for $i \in [\ell]$ and $z'(\ell+1) = z_L$ and similarly $y'(i)$ as $y(i)$ for $i \in [k]$ and $y'(k+1) = z_R$.

\begin{figure}[ht!]
    \centering
    \includegraphics[width=\textwidth]{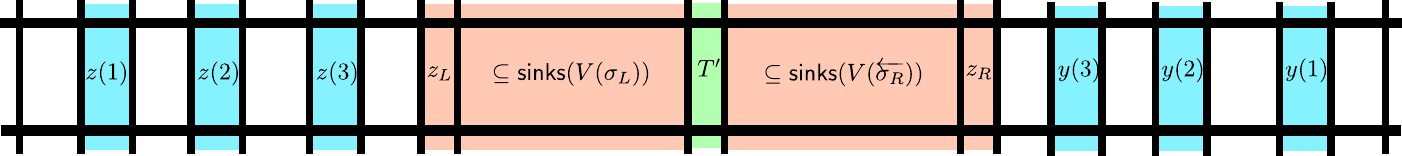}
    \caption{Definition of $z(i),y(i),z_L$ and $z_R$. Here $\ell$ and $k$ are
        equal $3$. Sink moments of $\sigma_L$
    and $\dual{\sigma_R}$ are highlighted blue. Green is highlighted the moment
    $T'$. In timeslots $[z_L,T'-1]$ and $[T'+1,z_R]$ only sinks of $\sigma_L$ and
    $\dual{\sigma_R}$ are scheduled.}
    \label{fig:zlzr}
\end{figure}

\subparagraph*{Case 1:}
Let $v \in (\SinkVC_L\setminus \EarlyVC_L)$, in other words, $v \in C$, $v \in \Sinks(\sigma_L)$ and $v$ is not early. For such a $v$ we take $l_v = z'(\depth{\HighVC_L}{v}+2)$ and $r_v = T'-1$.
It is easy to see that all successors of $v$ are processed after $r_v$: $v$ is a sink in $\sigma_L$, so it has no successors in $\sigma_L$ and (2) follows. 

Since $\sigma_L$ is sink-adjusted, we know that at any sink moment $t$ of $\sigma_L$ it holds that $S_{[t+1,T'-1]} \subseteq \Succ(S_{t}) \cup \Sinks(V(\sigma_L))$. Also, any chain can contain at most one vertex from $V\setminus C$ (Claim~\ref{claim:chain}). Hence after $z'(\depth{\HighVC_L}{v}+2)$ all predecessors of $v$ must be processed and (1) is indeed true.      

By definition of earliness, $v$ is not processed at the $(\depth{\HighVC_L}{v}+1)$th sink moment of $\sigma_L$. Additionally, $v$ cannot be processed at a sink moment before $z'(\depth{\HighVC_L}{v}+1)$, as at this sink moment its predecessors from $\HighVC_L$ are processed. Thus, since $v$ is processed in $\sigma_L$, (3) follows as well.

For $v \in (\SinkVC_R\setminus \EarlyVC_R)$ we define $l_v$ and $r_v$ in a
similar way, using the properties of $\sigma_R$.

\subparagraph*{Case 2:}
If $v \in U = V \setminus (C \cup I_L \cup I_R)$, the definition of $l_v$ and $r_v$ is a bit less straightforward. We do this by defining four possible lower bounds.  For notational simplicity, we let $\max \{\emptyset\} = 0$. 
\begin{align*}
    l^1_v &= \max\{ z'(i) : \exists u \in (C_L \setminus(\HighVC_L \cup \SinkVC_L))\cap \Pred(v) \text{ in $i$th block of $\sigma_L$}   \},\\ 
    l^2_v &= \max\{ z(i)+1 : \exists  u \in \HighVC_L\cap \Pred(v) \text{ in $i$th sink moment of $\sigma_L$} \},\\
    l^3_v &= \begin{cases} T' & \text{if } v \in \Succ(\SinkVC_L),  \\0 & \text {else,}
    \end{cases} \qquad
    l^4_v = \begin{cases} T'+1 & \text{if } v \in \Succ(C_M), \\0 & \text {else.}
    \end{cases}
\end{align*}
Similarly, for $r_v$ we define four upper bounds.  
\begin{align*}
    r^1_v &= \min\{ y'(i) : \exists u \in (C_R \setminus(\HighVC_R \cup \SinkVC_R))\cap \Succ(v) \text{ in $i$th block of $\dual{\sigma_R}$}   \},\\ 
    r^2_v &= \min\{ y(i) -1 : \exists  u \in \HighVC_R \cap \Succ(v)  \text{ in $i$th sink moment of $\dual{\sigma_R}$ \}},\\
    r^3_v &= \begin{cases} T' & \text{if } v \in \Pred(\SinkVC_R),  \\0 & \text {else,}
    \end{cases} \qquad
    r^4_v = \begin{cases} T'-1 & \text{if } v \in \Pred(C_M), \\0 & \text {else.}
    \end{cases}
\end{align*}
We then take $l_v = \max\{l^1_v, l^2_v, l^3_v, l^4_v\}$ and $r_v = \min\{r^1_v, r^2_v, r^3_v, r^4_v\}$.     
Note that the values of $l_v$ and $r_v$ can clearly be computed in polynomial
time, as they are simple expressions that only depend on $f$ and $G$. See
Figure~\ref{fig:inequalities} for schematic overview of lower and upper bounds.

\begin{figure}[ht!]
    \centering
    \includegraphics[width=\textwidth]{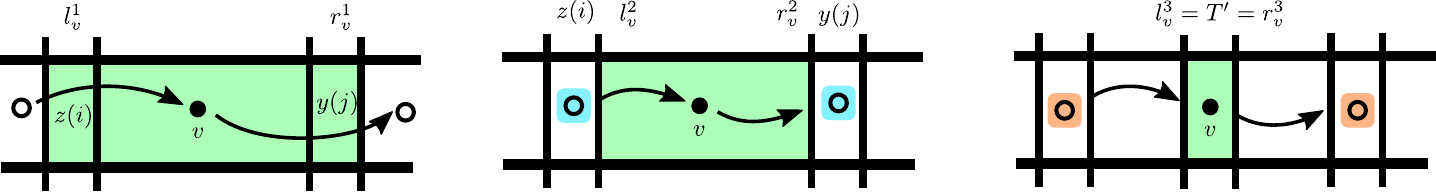}
    \caption{Schematic picture of determining lower bounds $l_v^1,l_v^2,l_v^3$
        and $r_v^1,r_v^2,r_v^3$. We highlighted green the available intervals
        (e.g., $[l_v^1,r_v^1]$) of job $v$. The first schema determines $l_v^1$ and
        $r_v^1$. For example, if vertex $u$ is in block $B_i$ then $l_v^1 \ge z(i)$.
        Middle schema says that if $v$ has predecessor from $\HighVC$ in sink-moment
        $z(i)$ then $l_v^2 > z(i)$. Last condition simply says that if a sink in
        $V(\sigma_L)$ is predecessor of $v$ then it needs to be processed at $T'$ or
        later. Inequalities for $l_v^4$ and $r_v^4$ are similar to the last
        figure.}
    \label{fig:inequalities}
\end{figure}

First we prove (1), the proof of (2) is similar. Let $u \in \Pred(v)\setminus U'$, as a
consequence $v \in \Succ(u)$. Because $v \not \in C$ we know $u \in C$.  The
vertex $u$ cannot be in $C_R$, as then we would have $v \in \Succ(C_R)$, i.e. $v
\in I_R$ and thus $v \not \in U'$.  If $u \in C_M$, then $u$ is processed at
$T'$ and before $l^4_v$.  If $u \in C_L \setminus(\HighVC_L\cup \SinkVC_L)$, $u$
cannot be processed at a sink moment of $\sigma_L$. If $u$ is processed at some
$i$th block of $\sigma_L$ for $i < \ell$, it is therefore always processed
before the $i$th sink moment because of bound $l_v^1$. If $u$ is processed at
the $(\ell+1)$th block of $\sigma_L$, then it is definitely processed before
$z'(\ell+1) = z_L$ as it is not a sink in $\sigma_L$. Therefore, it is also
processed before $l_v^1$.   If $u \in \HighVC_L$, then $u$ is processed at a
sink moment and before $l^2_v$.  If $u \in \SinkVC_L$, then $u$ is processed in
$\sigma_L$ and therefore before $l^3_v$.

For (3); we have to prove that $v$ is scheduled in interval $[l_v,r_v]$ in
$\sigma$. We show that $u$ is processed at or after $l_v$.  To this end, it is
sufficient to show that $u$ is processed after all lower bounds $l^1_v$,
$l^2_v$, $l^3_v$ and $l^4_v$ separately.  For $l^1_v$; if there is some $u \in
(C_L \setminus(\HighVC_L \cup \SinkVC_L))\cap \Pred(v)$ at the $i$th block of
$\sigma_L$, then it is processed somewhere strictly before $z'(i)$ as it is not a sink of $\sigma_L$. Because $v$ must be processed at a sink moment or after $z_L$,
it is processed at or after $z'(i)$ in $\sigma$.  For $l^2_v$; if there is some $u \in
\HighVC_L \cap \Pred(v)$ at the $i$th sink moment, $v$ is processed
after at some sink moment after $z(i)$ or after $z_L$.  If $l^3_v = T'$,  then  there is some $u \in \SinkVC_L$
such that $u \prec v$. Because $u$ is by definition a sink in $\sigma_L$, $v$
cannot be processed in $\sigma_L$. Therefore, $v$ is processed at or after $T'$.
If $l^4_v = T'+1$, then there is some $u \in C_M$ such that $u \prec v$.
Clearly, $v$ has to be processed at or after $T'+1$.  Hence $v$ is processed
after of at $l_v$. The proof that $u$ is processed before or at $r_v$ is
similar. This concludes the proof of Items (1-3).

It remains to show that condition $r_u < l_v$ holds if $u \prec v$ for every
$u,v \in U'$. Let $u,v \in U'$ and $u \prec v$. At least one of $u$ or $v$ is in
$C$. Recall that any job from $U'$ in $\sigma_L$ is a sink in $\sigma_L$ and any
job in $U'$ is $\sigma_R$ is a sink in $\dual{\sigma_R}$.  Therefore, when $u$
and $v$ are both in $C$, then $u\in \SinkVC_L$ and $v \in \SinkVC_R$ and by
definition $l_u < r_v$.  Now, let us assume that  $u \in C$ and $v \not\in C$
(the proof is analogous when $u \notin C$ and $v \in C$). Then $u$ cannot be in
$\SinkVC_R$ as $u \in \SinkVC_R$ and $u \prec v$ would imply $v \in I_R$ and
thus $v \not \in U'$. So, $u \in \SinkVC_L$ and $r_u = T'-1$. Because $u \in
\SinkVC_L$ and $u \in U'$, by definition then $l_v \ge l^3_v = T' > r_u$. 

Note that if $u \in U'$, $v \in C_M$ and $u\prec v$ then $r_u < T'$ and
similarly if $u \in U'$, $v \in C_M$ and $v\prec u$ then $T' < l_u$, because of
the lower and upper bounds $l^4_v$ and $r^4_v$.
\end{proof}

Given these $l_v$ and $r_v$ for each $v \in U'$, we add an edge $(v,(t,j))$ to $F$ if and only if $(t,j) \in P$ and $l_v \le t \le r_v$. 

\subparagraph*{The algorithm $\mathtt{divide}$.}

We will now finish the proof of Lemma~\ref{claim:reconstruction} by giving the
algorithm $\mathtt{divide}$ in Algorithm~\ref{algorithm2} and proving that it
has all properties of Lemma~\ref{claim:reconstruction}.

\begin{algorithm}[h!] 
\nonl\noindent\textbf{Algorithm} $\mathtt{divide}(G,m,T',C,f)$\\
\SetAlgoLined
\DontPrintSemicolon
$I_L \gets \Pred(C_L)\setminus C_L$, $I_R \gets \Succ(C_R)\setminus C_R$, $U \gets (V\setminus (C \cup I_L \cup I_R)$.\\
$U' \gets  U \cup (\SinkVC_L\setminus \EarlyVC_L) \cup (\SinkVC_R \setminus \EarlyVC_R) $\\
Compute $P$ and $F$ \tcp*{as discussed in Section~\ref{subsec:divide}}
Construct bipartite graph $H = ((P,U'),F)$\\
$\mathcal{M} \gets \mathtt{MaximumMatching}(H)$\\
\If{$\mathcal{M}$ is a perfect matching}{ 
$X_L \coloneqq I_L \cup \{v \in U: \{v,(t,j)\}\in \mathcal{M}, t \in [1,T'-1] \}$ \\
$X_M \coloneqq \{v \in U: \{v,(t,j)\}\in \mathcal{M}, t = T' \}$ \\
$X_R \coloneqq I_R \cup \{v \in U: \{v,(t,j)\}\in \mathcal{M}, t \in [T'+1,T] \}$ \\
\Return $(X_L,X_M,X_R)$}
\Return False
\caption{Algorithm for Lemma~\ref{claim:reconstruction}.}
\label{algorithm2}
\end{algorithm}
    
Clearly, $\mathtt{divide}$ runs in polynomial time as it construct graph $H$
using Claims~\ref{claim:blocks}~and~\ref{claim:intervals} (which both take
polynomial time) and then computes a perfect matching of $H$.  We are left to
show that if
$f=(C_L,C_R,\SinkVC_L,\SinkVC_R,\HighVC_L,\HighVC_R,\EarlyVC_L,\EarlyVC_R)$ is
the fingerprint of a tight $m$-schedule $\sigma$ of $G$ that is middle-adjusted
at time $T'$, then the partition $X_L,X_M,X_R$ of $V \setminus C$  returned by
$\mathtt{divide}$ has the following properties: $G[C_L \cup X_L]$ and $G[C_R
\cup X_R]$ have tight $m$-schedules, $|X_M \cup C_M| = m$, $\Pred(C \setminus C_R) \subseteq C_L \cup
X_L$, and $\Succ(C \setminus C_L) \subseteq C_R \cup X_R$.

First, we prove that $\mathtt{divide}$ returns a partition at all. In other
words, we show that the bipartite graph $H$ has a perfect matching. We claim
there is a perfect matching of $H$ based on $\sigma$. By matching vertices to
any position at the time slot they are processed in $\sigma$, we get a perfect
matching. These edges must exist in $H$ because of (3) in
Claim~\ref{claim:intervals}. 

Because by construction there are $m -|C_M|$ position in $P$ with $t = T'$ and $\mathcal{M}$ is a perfect matching, $|X_M \cup C_M| = m$.

Next, we prove $G[C_L \cup X_L]$ has a tight $m$-schedule. Take $\sigma_L$ and
remove any jobs from $U'$. This leaves exactly the positions in the set $P$ to
be empty by Claim~\ref{claim:blocks}. Then construct the schedule $\sigma_L'$ by
processing each job $v \in U'$ at the timeslot a job $v$ is matched to in the
matching $\mathcal{M}$.  More precisely, let $v \in U'$ be matched to some
position $(t,j)$ for $t<T'$ by $\mathcal{M}$, then process $v$ at time $t$ in
$\sigma_L'$.  Because of properties (1-2) of Claim~\ref{claim:intervals}, we
know that all jobs in $\Pred(v)\setminus U'$ are scheduled before $l_v \le t$
and all jobs in $\Succ(v) \setminus U'$ are processed after $r_v \ge t$.
Furthermore, if there is some $u \in U'$ that is comparable to $v$, then we know
that their intervals imply the precedence constraints. Finally, since
$\mathcal{M}$ is a perfect matching, all positions are filled. Hence, we have a
tight $m$-schedule.  With similar arguments $G[C_R \cup X_R]$ has a tight
$m$-schedule.

It remains to show $\Pred(C \setminus C_R) \subseteq C_L \cup X_L$. Take $v \in
\Pred(C \setminus C_R) = \Pred(C_L \cup C_M)$. If $v \in \Pred(C_L)$ then $v \in
C_L \cup I_L \subseteq C_L \cup X_L$. If $v \in \Pred(C_M)$, then by
Claim~\ref{claim:intervals} we have $r_v < T'$ and so $v \in C_L\cup X_L$.  Similarly we
can show that $\Succ(C \setminus C_L) \subseteq C_R \cup X_R$. 

This concludes the proof of Lemma~\ref{claim:reconstruction}.

\section{Getting below $2^n$: Proof of Theorem~\ref{thm:mainthm}}
\label{sec:exact} 
In this section we give the present the two exact algorithms 
needed to prove our main result,  Theorem~\ref{thm:mainthm}. 
We first give an $\Os (2^n)$ time algorithm using Fast Subset Convolution for \sched in 
Subsection~\ref{subsec:subsetconvolution}. We then improve this result and give an $\Os(\AC+2^{n-m})$ 
algorithm in Subsection~\ref{subsec:fastersubset}. In Subsection~\ref{subsec:DP} we present a natural Dynamic 
Programming algorithm that runs in $\Os(\AC\binom{n}{m})$. In Subsection~\ref{subsec:proofmain} we prove that these algorithms together with Theorem~\ref{thm:VC} can be combined into an algorithm solving \sched in $\Oh(1.995^n)$ time.

\subsection{An $\Os(2^n)$ algorithm using Fast Subset Convolution} \label{subsec:subsetconvolution}

In this subsection, we show how to use Fast Subset Convolution to solve \sched
in $\Os(2^n)$ time. 

\begin{theorem}
    \label{thm:baseline}
    \sched can be solved in time $\Os(2^n)$. 
\end{theorem}

This is a base-line of our methods.  Later, we will then use this algorithm to
get a faster than $\Os(2^n)$ algorithm in the case $m \ge n/236$ in
Subsection~\ref{subsec:fastersubset}. To prove~\cref{thm:baseline} let us first recall what we can do with
fast subset convolutions.

\begin{theorem}[Fast subset convolution with Zeta/ M\"obius transform~\cite{BjorklundHKK09}]
    \label{Thm:zetamob}
    Given functions $f,g: 2^U \rightarrow \mathbb{N}$. There is an algorithm that
    computes 
    \begin{displaymath}
        (f \circledast g)(S) \coloneqq \sum_{T \subseteq S} f(T) \cdot g(S \setminus T)
    \end{displaymath} 
    for every $S \subseteq U$ in 
    $2^{|U|}\cdot |U|^{\Oh(1)}$
    ring operations.
\end{theorem}

We will use this convolution multiple times in our algorithm.
The plan is to encode the set of jobs $V$ as the universe $U$. Then
the $f$ function will encode whether it is possible to process the jobs of $X \subseteq V$
within a given time frame. Function $g$ will be used to check whether the set of jobs $Y
\subseteq V$can be processed at the last time-slot. We define these function formally.

For any $X \subseteq V$ and $t \in [T]$ let
\begin{displaymath}
    f_{t}(X) \coloneqq \begin{cases} 1 & \text{if jobs $\Pred[X]$ can be
    processed within first $t$ time slots, and $X = \Pred[X]$}, \\ 0 & \text{otherwise.} \end{cases}
\end{displaymath}

For any $Y \subseteq V$ define 
\begin{displaymath}
    g(Y) \coloneqq \begin{cases} 
        1 & \text{ if } |Y| \le m \text{ and } Y \text{ is an antichain }, \\
        0 & \text{ otherwise,}
    \end{cases}
\end{displaymath}

Note that the value $f_{T}(V)$ tells us whether the set of jobs can be processed
within $T$ time units and is therefore the solution to our problem.
Additionally, observe that the base-case $f_{0}(X)$ can be efficiently determined
for all $X \subseteq V$ because $f_{0}(X)=1$ if $X = \emptyset$ and $f_{0}(X)=0$
otherwise. Moreover, for a fixed $Y \subseteq V$, the value of $g(Y)$ can be
found in polynomial time. 

It remains to compute $f_{t}(X)$ for every $t>0$. To achieve this, we define an
auxiliary function $h_{t} : 2^V \rightarrow \nat$. For every $Z \subseteq V$, let
\begin{displaymath}
    h_t(Z) \coloneqq \sum_{X \subseteq Z} f_{t-1}(X) \cdot g(Z \setminus X)
    .
\end{displaymath}
Once all values of $f_{t-1}(X)$ are known, the values of $h_{t}(Z)$ for $Z
\subseteq V$ can be computed in time $\Oh(2^{n})$ time using
Theorem~\ref{Thm:zetamob}. Next, for every $X \subseteq V$ we determine the
value of $f_t(X)$ from $h_t(X)$ as follows:
\begin{displaymath}
    f_{t}(X) = \iverson{h_{t-1}(X)\ge 1} \cdot \iverson{X = \Pred[X]}. 
\end{displaymath}

For every $X \subseteq V$ this transformation can be done in polynomial time.
Therefore, the total runtime of computing $f_t$ is 
$2^n\cdot n^{\Oh(1)}$
To prove
correctness of our algorithm and finish a proof of~\cref{thm:baseline}, it
suffices to prove the following lemma:

\begin{lemma}[Correctness]\label{lem:Zetacorr}
    Let $X,Z \subseteq V$ be such that $Z = \Pred[Z]$ and let $Y \coloneqq X
    \setminus Z$. Then, the following statements are equivalent:
    \begin{itemize}
        \item $X = \Pred[X]$ and $Y$ is an antichain.
        \item $Y \subseteq \Sinks(Z)$.
    \end{itemize}
\end{lemma}
\begin{proof}
    \textbf{($\Uparrow$)}:
    Assume that $Y \subseteq \Sinks(Z)$. Then automatically $Y$ is an antichain. It
    remains to check that for all $v \in X$ it holds that $\Pred(v) \subseteq X$. Because $\Pred(v) \subseteq \Pred[Z] = Z$ and $Y$ contains only sinks. This
    means that $X = \Pred[X]$.

    \textbf{($\Downarrow$)}: Assume that $Y$ is an antichain and $X = \Pred(X)$.
    Take any $v \in Y$ and assume $v \not \in \Sinks(Z)$. However then there is
    a successor $v' \in Z$ of $v$, i.e., $v \prec v'$. However $Y$ is an
    antichain and $v' \not \in Y$.  Hence it must be that $v'\in X$. But then
    $\Pred(v') \not \subseteq X$, which contradicts the that $\Pred[X] = X$.
\end{proof}

This concludes the proof of~\cref{thm:baseline}. Note, that the above algorithm
computes all the values of dynamic programming.

\begin{remark}
    \label{rem:stronger-baseline}
    Given an instance of \sched, we can compute in $\Os(2^n)$ time the value of $f_t(X)$ for every $t \in [T]$ and
    $X \subseteq V$.
\end{remark}

\subsection{An $\Os(2^{n-m} + \AC)$ algorithm for Theorem~\ref{thm:baseline2}} \label{subsec:fastersubset}

Now, we will use Theorem~\ref{thm:baseline} as a subroutine and show that \sched can be
solved in $\Os(2^{n-m} + \AC)$ time.

\begin{theorem}
    \label{thm:baseline2}
    \sched can be solved in time $\Os(2^{n - m} + \AC)$. 
\end{theorem}

As usual, we use $T$ to denote the makespan.  First, we assume that $n \le m T$
because otherwise the answer is trivial \emph{no}. Our algorithm uses the
following reduction rules exhaustively.

\begin{reduction} 
    \label{red:isolated}
    Remove every isolated vertex from the graph.
\end{reduction}

\begin{reduction} 
    \label{red:sources}
    If there are $\le m$ sources (or $\le m$ sinks), we remove these sources
    (or sinks) from the graph and decrease $T$ by $1$. 
\end{reduction}

For the correctness of~\cref{red:isolated} assume that a schedule after
application of~\cref{red:isolated} has $n'$ jobs and makespan $T$.  It means
that the schedule has $mT - n'$ available slots. We can schedule the deleted
jobs at any these slots because these jobs do not have any predecessor and
successor constraints. The makespan of the schedule remains $T$, because we
assumed that the initial number of jobs is $\le mT$.

For correctness of~\cref{red:sources} observe that if a dependency graph has
$\le m$ sources then there exists an optimal schedule that processes these
sources at the first time slot. By symmetry if the dependency graph has $\le m$
sinks then in some optimal schedule these sinks are processed at the last
timeslot.  Moreover only sources can be processed at the first timeslot and only
sinks can be processed at the last timeslot.

Therefore, we may assume that there are at least $m$ sources and at least $m$
sinks in the dependency graph and there are no isolated vertices in the
dependency graph. Now, let us use Theorem~\ref{thm:baseline} as a subroutine.

Let $\sigma$ be an optimal sink-adjusted schedule and let $z$ be the first
moment a sink is processed in $\sigma$. 
By definition of sink-adjusted schedule either $z$ is a sink moment, or
$S_{[z,T]} \subseteq \Sinks(V)$. We may assume that no sources are
processed after $z$; otherwise we could switch the sink at time $z$ with such a
source (observe that by~\cref{red:isolated} we know that no job can be source
and sink at the same time).


Now we use~\cref{rem:stronger-baseline} and compute the values of $f_t(X)$ for all
$X\subseteq(V \setminus \Sinks(V))$ and $t \in [T]$ on graph $G[ V \setminus
\Sinks(V)]$.  Observe that graph $G \setminus \Sinks(V)$ contains at most $n-m$
jobs.  Therefore, computing all these values takes $\Os(2^{n-m})$ time.  Next, we
take graph $G[V \setminus \Sources(V)]$. We reverse all its arcs and
use~\cref{thm:baseline2} to compute the values $\dual{f}_t(X)$ for all
$X\subseteq (V \setminus \Sources(V))$ and $t \in [T]$ in time $\Os(2^{n-m})$.

After this preprocessing, we guess set $S_z \subseteq V$. Observe that the jobs
in $S_z$ form an antichain. Moreover we can enumerate all the anti-chains of
$G$ in $\Ot(\AC)$ time with the following folklore algorithm: start with a
minimal anti-chain.  Then guess the next vertex that you want to add to to your
current anti-chain and remove all the elements that are comparable to the
guessed vertex. Finally add the current anti-chain to your list and branch on
the next element.
In total, in order to guess $S_z$ and to compute functions $f_t$ and
$\dual{f_t}$ we need $\Os(2^{n-m} + \AC)$ time. It remains to argue that with $S_z$, $f_t$ and
$\dual{f_t}$ in hand we can solve \sched in polynomial time.
First, recall that $z$ is either the first sink moment or
$S_{[z,T]} \subseteq \Sinks(V)$. 
This means that we can identify set $S_{[z+1,T]} \coloneqq (\Succ({S_z}) \cup
\Sinks(V))\setminus S_z $ of jobs processed after $z$. Similarly, we can deduce 
set $S_{[1,z-1]} \coloneqq V \setminus S_{[z,T]}$ of jobs that are processed
before $z$. It remains to verify (by inspecting the
functions $f_{z-1}$ and $\dual{f}_{T- z-1}$) that jobs $S_{[1,z-1]}$ can be
processed in the first $z-1$ timeslots and jobs $S_{[z+1,T]}$ can be processed
within the  $T-z-1$ last timeslots. This concludes the description of the algorithm
and proof of Theorem~\ref{thm:baseline2}.



\subsection{An $\Os(\#AC\cdot \binom{n}{m})$ algorithm using Dynamic Programming}
\label{subsec:DP}

The natural Dynamic Program for the problem is as follows. We emphasize that this algorithm is folklore (for example it was also mentioned in~\cite{JansenLK16} and~\cite{nederlof2022fine}).

\begin{theorem}\label{thm:DP}
Let $\AC$ denote the number of different antichains of $G$. Then \sched can be solved in time $\Os( \AC \cdot \binom{n}{m})$.
\end{theorem} 

\begin{proof} 

    Our algorithm is based on dynamic programming.
    For every antichain $B \subseteq V$ of graph $G = (V,A)$ and integer $t \in [0,T]$ we define the
    states of dynamic programming $\DP_t[B]$ as follows:
    \begin{displaymath}
        \DP_t[B] \coloneqq \begin{cases} 1 &\text{ if jobs of }
        \Pred[B] \text{ can be scheduled within the first $t$ timeslots}, \\ 0 &\text{ otherwise}.
    \end{cases}
    \end{displaymath}
    Clearly, $\DP_0[\emptyset] =1$ and $\DP_0[B] = 0$ for any nonempty antichain $B$. 
    We use the following recurrence relation to compute the subsequent entries of dynamic
    programming table for every $t$ from $1$ to $T$: 

    \begin{displaymath}
        \DP_{t}[B] = \max_{X \subseteq B} \Big( \DP_{t-1}[\Sinks(\Pred[B] \setminus X)] \Big).
    \end{displaymath}

    We show correctness of the recurrence above.
    First, note that $\Sinks(\Pred[B] \setminus X)$ is always an antichain as it is a set of sinks, which are by definition incomparable. Furthermore, $\Pred[\Sinks(\Pred[B] \setminus X)] = \Pred[B]\setminus X$. 
    Now assume $\sigma$ is a schedule that processes the jobs in $\Pred[B]$ of 
    makespan $t$. Then at time $t$ the only jobs from $\Pred[B]$ that can be 
    processed are the jobs from $B$ itself; they are the sinks of $\sigma$. 
    Let $X = S_t$, then there is a schedule $\sigma'$ that can process 
    $\Pred[B]\setminus X$ in time $t-1$. Hence, $\DP_{t-1}[\Sinks(\Pred[B] 
    \setminus X)]=1$ and so $\DP_{t}[B]=1$. 
    
    For the other direction, assume that for an antichain $B$, $X\subseteq B$ 
    and $t \in [T]$ we find $\DP_{t-1}[\Sinks(\Pred[B] \setminus X)]=1$. Then 
    we also find that there is a schedule for $\Pred[B]$ with makespan $t$: 
    take the schedule for $\Pred[B] \setminus X$ and process $X$ at timeslot $t$. 
    Because $B$ is an antichain, all jobs in $X$ are incomparable. 
    Furthermore, all predecessors of jobs in $X$ were already processed before 
    $t$. This concludes the proof of correctness.

    As for the runtime, observe there there are $\Oh(n \cdot \AC)$ entries in the
    table $\DP$ and number of possibilities for $X \subseteq B$ is
    $\binom{n}{m}$. Moreover all the anti-chains of $G$ can be computed in
    $\Ot(\AC)$ time (see Section~\ref{sec:exact}).
\end{proof}

\subsection{Combining all parts}
\label{subsec:proofmain}
It remains to prove Theorem~\ref{thm:mainthm}, i.e. give an algorithm that solves \sched in $\Oh(1.995^n)$ time. To do this, we first need the following claim that follows from Dilworth's Theorem. 

\begin{claim}\label{claim:nrantichains}
    Let $G$ be a poset with $n$ vertices. If the minimum vertex cover of its
    comparability graph $\Gc$ has size at least $(1-\alpha) n$ for some constant $\alpha \in (0,1)$,
    then
    \begin{displaymath}
        \AC(G) \le \left( 1 + \frac{1}{\alpha} \right)^{\alpha n}
        .
    \end{displaymath}
\end{claim}
\begin{proof}
    Assume that the size of minimum vertex cover of $\Gc$ is at least
    $(1-\alpha) n$. By duality, $\Gc$ has an maximum independent $I$ set of size
    at most $\alpha n$.  Because there are no edges in $\Gc[I]$, the set $I$ is
    an antichain in $G$.  Next, we use the Dilworth's
    Theorem~\cite{dilworth2009decomposition} that states the graph $G$ can be
    decomposed into $\ell \le |I| = \alpha n$ chains $C_1,\dots,C_\ell$.  

    Observe that every antichain can be succinctly described by either (i)
    selecting one of its vertex, or (ii) deciding to select none. Hence
    $\AC(G) \le \prod_{i=1}^{\ell} (|C_i|+1)$.
    Next, we use the AM-GM inequality. We get that:
    \begin{displaymath}
        \prod_{i=1}^{\ell} (|C_i|+1) \le \left(\frac{\sum_{i=1}^\ell (|C_i| + 1)}{\ell}\right)^\ell
    \end{displaymath}
    Observe that $\sum_{i=1}^\ell |C_i| = n$. Hence
    $\AC(G) \le (n/\ell+1)^{\ell} \le \left( 1+ \frac{1}{\alpha} \right )^{\alpha n}$.
\end{proof}

We note that Claim~\ref{claim:nrantichains} is tight, as $G$ could simply
consist of $\alpha n$ chains each of length $1/\alpha$.  

We are now ready to prove our main Theorem. See Figure~\ref{tikz:MainThmProof} for an overview of
the algorithm.

\begin{proof}[Proof of Theorem~\ref{thm:mainthm}]
First, we compute
the vertex cover $C$ of the comparability graph. This step can be done in
$\Os(1.3^n)$ (see~\cite{chen2010improved}). 

If $|C|\le \frac{n}{7.5}$, we observe that Theorem~\ref{thm:VC} is
fast enough as $169^{|C|} < 1.995^n$. 
Hence 
we can assume that the vertex cover is large, i.e. $|C|> \frac{n}{7.5}$. Claim~\ref{claim:nrantichains} then guarantees that the number of antichains is $\AC \le \Oh(1.9445^n)$.
For that case, we propose two algorithms based on the number of machines. 

When
the number of machines $m\le n/258$, we use the standard the dynamic programming from Subsection~\ref{subsec:DP} that runs in $\Os(\AC \cdot \binom{n}{m})$ time. As for $m \le n/258$, we can bound $\binom{n}{m} \le 1.0257^n$, we find that this is fact enough.  

In the
remaining case $m > n/258$, we apply the modified Fast Subset
Convolution algorithm described in Subsection~\ref{subsec:fastersubset}, running in
$\Os(\AC + 2^{n-m})$. This is fast enough because $m > n/258$.
This concludes the proof.

\end{proof}

\section{Conclusion and Further Research}
\label{sec:conc}

In this paper, we analyse \sched from the perspective of exact exponential time
algorithms. We break the $2^n$ barrier 
by presenting a $\Oh(1.995^n)$ time algorithm for \sched. This result is based
on a tradeoff between the number of antichains of the input graph and the size
of the vertex cover of its comparability graph. Our main technical contribution
is a $\Os(169^{|C|})$ time algorithm where $C$ is a vertex cover of the
comparability graph. To achieve this, we extend the techniques introduced by Dolev and Warmuth~\cite{dolev1984scheduling}. 

It would be interesting to improve our main theorem for a fixed number of
machines. Since $Pm|\text{prec}, p_j=1 | C_{\max}$ is not known to be $\mathsf{NP}$-complete for fixed $m$, 
one might even aim for subexponential time algorithms.
Even for $m=3$, this would be a breakthrough. 

We note that fixed-parameter tractable algorithms for non-trivial parameterizations are rare in the field of scheduling problems (see, e.g., survey by~\cite{mnich2018parameterized}).
The constant $169$ in the base of the exponent is relatively large and any
improvement to it would ultimately lead to a faster algorithm for \sched. We
believe that even reducing the runtime below $\Os(10^{|C|})$ requires a
significantly new insight into the problem. Note however that even if one could
somehow assume that $\AC \approx 1.1^n$ the current best algorithms from
Section~\ref{sec:exact} would guarantee only  $\Oh(1.993^n)$ time algorithm.
To improve our algorithm below $\Oh(1.9^n)$ one likely needs completely new
ideas.



Another interesting approach would be to find fixed-parameter tractable algorithms for other parameters. One such parameter is $h$, the height of the input graph. Even for fixed height, \sched is $\mathsf{NP}$-hard. However, for fixed number of machines, the problem is in $\mathsf{XP}$ when parameterized by the height, thanks to the algorithm of Dolev and Warmuth~\cite{dolev1984scheduling}. We wonder whether a fixed-parameter tractable algorithm is also possible, even for $m=3$.

Finally, while there is ample evidence that no $2^{o(n)}$ time algorithm exists for \sched, it remains a somewhat embarrassing open problem to show that such an algorithm would violate the Exponential Time Hypothesis.

\begin{picture}(0,0)
\put(250,-30)
{\hbox{\includegraphics[width=40px]{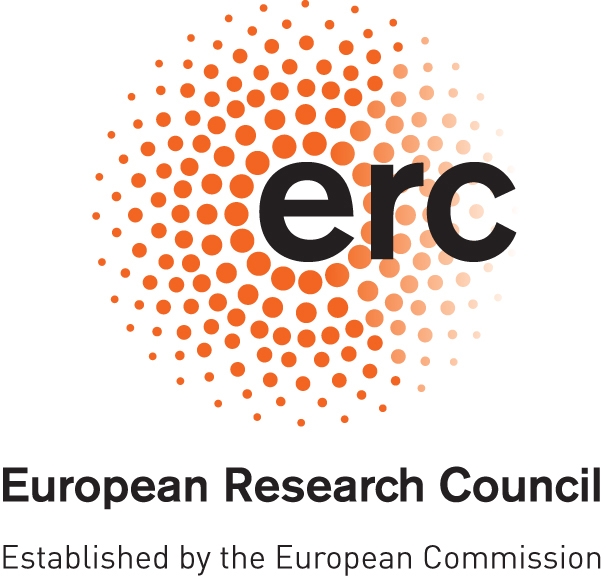}}}
\put(300,-50)
{\hbox{\includegraphics[width=60px]{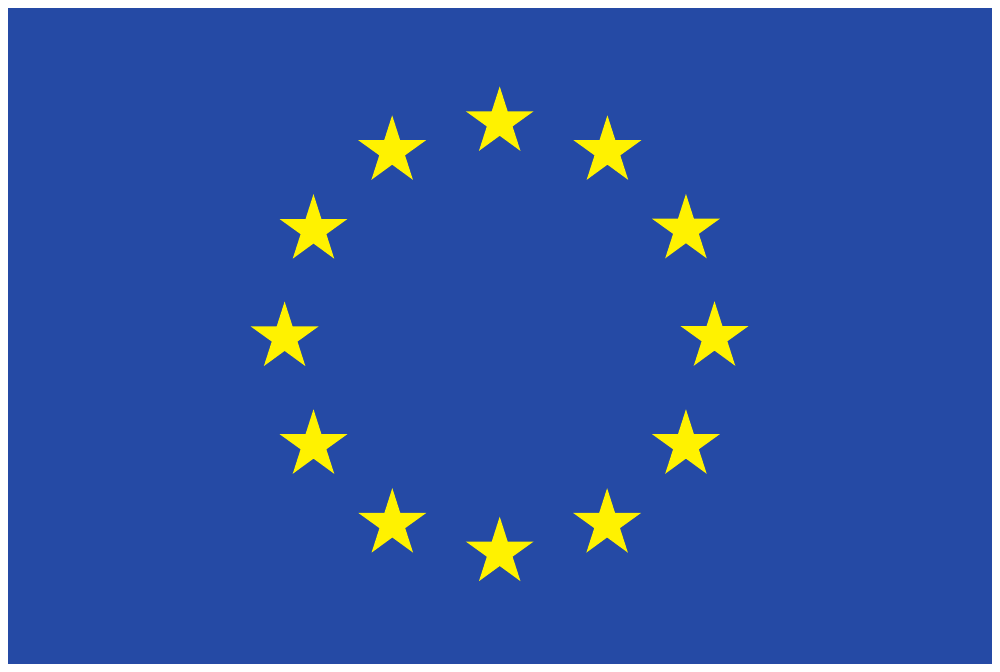}}}
\end{picture}





\bibliographystyle{plain}
\bibliography{references}
\appendix
\section{Lower Bound} 
\label{sec:LB}

Lenstra and Rinnooy Kan~\cite{lenstra1978complexity} proved $\mathsf{NP}$-hardness of \sched. They
reduced from an instance of \textsc{Clique}  with $n$ vertices to an instance of
\sched with $\Oh(n^2)$ jobs. Upon a close inspection their reduction gives
$2^{\Omega(\sqrt{n})}$ lower bound (assuming the Exponential Time Hypothesis). 
Jansen, Land and Kaluza~\cite{JansenLK16} improve this to $2^{\Omega(\sqrt{n \log n})}$. To the best of our knowledge this the currently best lower bound based on the Exponential Time Hypothesis. They also show that a $2^{o(n)}$ time algorithm for \sched would imply a $2^{o(n)}$ time algorithm for the Biclique problem on graphs on $n$ vertices.

We modify the reduction form~\cite{lenstra1978complexity} and start from an instance of \textsc{Densest
$\kappa$-Subgraph} on sparse graphs. 

In the \textsc{Densest $\kappa$-Subgraph} problem (\DKS), we are given a graph $G
= (V,E)$ and a positive integer $\kappa$. The goal is to select a subset $S
\subseteq V$ of $\kappa$ vertices that induce as many edges as possible. We use
$\den(G)$ to denote $\max_{S \subseteq V, |S| = \kappa} |E(S)|$, i.e. the optimum of
\DKS. Recently, Goel et al.~\cite{pasin} formulated the following Hypothesis
about the hardness of \DKS.

\begin{hypothesis}[\cite{pasin}]
    \label{hyp:dks}
    There exists $\delta > 0$ and $\Delta \in \nat$ such that the following holds.
    Given an instance $(G,\kappa,\ell)$ of
    \DKS, where each one of $N$ vertices of graph $G$ has degree at most $\Delta$,
    no $\Oh(2^{\delta N})$ time algorithm can
    decide if $\den(G) \ge \ell$.
\end{hypothesis}

In fact Goel et al.~\cite{pasin} formulated much stronger hypothesis about a
hardness of approximation of \DKS. \cref{hyp:dks} is a special case of
\cite[Hypothesis 1]{pasin} with $C = 1$.  Now we exclude $2^{o(n)}$ time
algorithm for \sched assuming~\cref{hyp:dks}. To achieve this we modify
the $\mathsf{NP}$-hardness reduction of~\cite{lenstra1978complexity}.

\begin{theorem}
    \label{thm:lb-sched}
    There is no algorithm that solves \sched in $2^{o(n)}$ time assuming~\cref{hyp:dks}.
\end{theorem}

\begin{proof}

    We reduce from an instance $(G,\kappa,\ell)$ of \DKS as
    in~\cref{hyp:dks}. We assume that graph $G$ does not contain isolated
    vertices (note that if any isolated vertex is part of the optimum solution to \DKS then
    an instance is trivial).
    We are promised that $G$ is $N$ vertices graph with $M \le \Delta N$ many
    edges for some constant $\Delta \in \nat$. Based on $(G,\kappa,\ell)$ we
    construct the instance of \sched as follows. 
    \begin{itemize}
        \item For each vertex $v \in V(G)$ create job $j_v^{(1)}$.
        \item For each edge $e = \{u,v\} \in E(G)$ create job $j_e^{(2)}$ with precedence constraints $j_u^{(1)} \prec j_e^{(2)}$ and  $j_v^{(1)} \prec j_e^{(2)}$.  
    \end{itemize}

    Next, we set the number of machines $m \coloneqq 2 \Delta N + 1$ and  create
    \emph{filler} jobs. Namely, we create three layers of jobs: Layer $L_1$
    consists of $m-\kappa$ jobs, layer $L_2$ consists of $m + \kappa - \ell - N$ jobs and
    layer $L_3$ consists of $m + \ell - M$ jobs. Finally, we set all the jobs
    in $L_1$ to be predecessors of every job in $L_2$ and all jobs in $L_2$ to
    be predecessors of $L_3$. This concludes the construction of the instance.
    At the end we invoke an oracle to \sched and declare that the $\den(G) \ge
    \ell$ if the makespan of the schedule is $T=3$.

    Now we argue that the constructed instance of \sched is equivalent to
    the original instance of \DKS.

    $\bold{(\Rightarrow)}$: Assume that an answer to \DKS is true and there
    exist set $S \subseteq V$ of $\kappa$ vertices that induce $\ge \ell$ edges.
    Then we can construct a schedule of makespan $3$ as follows. In the first
    timeslot take jobs $j^{(1)}_v$ for all $v \in S$ and all jobs from layer
    $L_1$. In the second timeslot take (i) jobs $j^{(1)}_u$ for all $v \in V
    \setminus S$, (ii) arbitrary set of $\ell$ jobs $j^{(2)}_e$ where $e =
    \{u,v\}$ and $u,v \in S$, and (iii) all the jobs from $L_2$. In the third
    timeslot take all the remaining jobs. Note that all precedence constraints
    are satisfied and the sizes of $L_1,L_2$ and $L_3$ are selected such that
    all of timeslots fit $\le m$ jobs.

    $\bold{(\Leftarrow)}$: Assume that there exists a schedule with makespan
    $3$. Because the total number of jobs $n$ is $3m$ every timeslot must be full,
    i.e., exactly $m$ jobs are scheduled in every timeslot.  Observe that jobs
    from from layers $L_1,L_2$ and $L_3$ must be processed consecutively in
    timeslots $1$, $2$ and $3$ because every triple in $L_1 \times L_2 \times
    L_3$ forms a chain with $3$ vertices.  Next, let $S \subseteq V$
    be the set of vertices such that jobs $j^{(1)}_s$ with $s \in S$ are processed in the
    first timeslot. Observe that (other than jobs from $L_1$) only $\kappa$ jobs
    of the form $j^{(1)}_v$ for some $v \in V$ can be processed in the first
    timeslot (as these are the only remaining sources in the graph). Now,
    consider a second timeslot. It must be filled by exactly $m$ jobs.  There is
    exactly $N - \kappa$ jobs of the form form $j^{(1)}_v$ for $v \in V \setminus S$
    and exactly $m - \ell - (N - \kappa)$ jobs in $L_2$. Therefore, $\ell$ jobs
    of the form $j^{(2)}_e$ for some $e \in E(G)$ must be scheduled in second
    timeslot. These jobs correspond to the edges of $G$
    with both endpoints in $S$.  Hence $\den(G) \ge \ell$.

    This concludes the equivalence between the instances. For the running time 
    observe that the number of jobs $n$ in the constructed instance is $3 m$.
    This is $\Oh(N)$ because $\Delta$ is constant. Hence an algorithm that runs in
    $2^{o(n)}$ time and solves \sched contradicts~\cref{hyp:dks}.
\end{proof}

\end{document}